%% file: main_prr.tex
\documentclass[12pt,american,reprint,prr,longbibliography,superscriptaddress]{revtex4-2}
\usepackage[T1]{fontenc}
\usepackage[latin9]{inputenc}
\usepackage{color}
\usepackage{babel}
\usepackage{amsmath}
\usepackage{amsthm}
\usepackage{amssymb}
\usepackage{graphicx}
\usepackage[unicode=true,pdfusetitle,
 bookmarks=true,bookmarksnumbered=false,bookmarksopen=false,
 breaklinks=false,pdfborder={0 0 0},pdfborderstyle={},backref=false,colorlinks=true]
 {hyperref}

\makeatletter
\theoremstyle{plain}
\newtheorem{ax}{\protect\axiomname}
\theoremstyle{definition}
\newtheorem{condition}{\protect\conditionname}
\theoremstyle{definition}
\newtheorem{defn}{\protect\definitionname}
\theoremstyle{plain}
\newtheorem{prop}{\protect\propositionname}
\theoremstyle{plain}
\newtheorem{thm}{\protect\theoremname}
\theoremstyle{definition}
 \newtheorem{example}{\protect\examplename}
\theoremstyle{plain}
\newtheorem{lem}{\protect\lemmaname}

\makeatother
\input{myQcircuit}
\makeatletter

\newcommand{\rA}{\mathrm{A}}
\newcommand{\rB}{\mathrm{B}}
\newcommand{\rC}{\mathrm{C}}
\newcommand{\rD}{\mathrm{D}}
\newcommand{\rS}{\mathrm{S}}
\newcommand{\rE}{\mathrm{E}}
\newcommand{\cA}{\mathcal{A}}
\newcommand{\cB}{\mathcal{B}}
\newcommand{\cC}{\mathcal{C}}

\newcommand{\cT}{\mathcal{T}}

\makeatother

\providecommand{\axiomname}{Axiom}
\providecommand{\conditionname}{Condition}
\providecommand{\definitionname}{Definition}
\providecommand{\examplename}{Example}
\providecommand{\lemmaname}{Lemma}
\providecommand{\propositionname}{Proposition}
\providecommand{\theoremname}{Theorem}

\begin{document}
\title{Universal structure of objective states in all fundamental causal
theories}
\author{Carlo Maria Scandolo}
\email{carlomaria.scandolo@ucalgary.ca}

\affiliation{Department of Mathematics \& Statistics, University of Calgary, T2N
1N4 Calgary, AB, Canada}
\affiliation{Institute for Quantum Science and Technology, University of Calgary,
T2N 1N4 Calgary, AB, Canada}
\author{Roberto Salazar}
\email{roberto.salazar@uj.edu.pl}

\affiliation{Institute of Informatics, National Quantum Information Centre, Faculty
of Mathematics, Physics and Informatics, University of Gda\'{n}sk,
80-308 Gda\'{n}sk, Poland }
\affiliation{Faculty of Physics, Astronomy and Applied Computer Science, Jagiellonian
University, 30-348 Kraków, Poland }
\author{Jaros\l aw K.\ Korbicz}
\affiliation{Center for Theoretical Physics, Polish Academy of Sciences, 02-668
Warsaw, Poland}
\author{Pawe\l{} Horodecki}
\affiliation{International Centre for Theory of Quantum Technologies, University
of Gda\'{n}sk, 80-308 Gda\'{n}sk, Poland}
\affiliation{Faculty of Applied Physics and Mathematics, National Quantum Information
Centre, Gda\'{n}sk University of Technology, 80-233 Gda\'{n}sk, Poland}
\begin{abstract}
A crucial question is how objective and classical behavior arises
from a fundamental physical theory. Here we provide a natural definition
of a decoherence process valid in all causal theories, and show how
its behavior can be extremely different from the quantum one. Remarkably,
despite this, we prove that the so-called spectrum broadcast structure
characterizes all objective states in every fundamental causal theory,
exactly as in quantum mechanics. Our results show a stark contrast
between the extraordinarily diverse decoherence behavior and the universal
features of objectivity.
\end{abstract}
\maketitle

\section{Introduction}

A common experience in our everyday life is that different observers
agree on their observations. This agreement means that macroscopic
physics is \emph{objective}. Note that this is a general feature of
classical physics, but it contrasts with quantum theory, where states
are generally disturbed by the act of observation, and sometimes an
agreement between observers is impossible \citep{Frauchiger-Renner,Brukner_objectivity}.
Nevertheless, in quantum mechanics there are objective states, in
the sense that various observers can determine them without disturbance
\citep{Zurek-objectivity}. It is argued that such objective states
may indeed be responsible for the objectivity we experience in our
everyday life \citep{Zurek-objectivity,Objectivity}. In quantum mechanics,
the theory of decoherence first \citep{Zeh-first,Zeh-second,Decoherence-review,Review-decoherence},
later quantum Darwinism \citep{QDarwinism,Zurek-objectivity,Witness,Darwinism-branches,BPH,Knott}
and the presence of the so-called spectrum broadcast states (SBSs)
\citep{Broadcasting-scenarios,J1PRL,Objectivity,JJ3,JJ4,JJ5,Miro,SBS-measurements,Strong-QDarwinism,Jarek-truth,QD-SBS}
have been proposed as explanations for the emergence of classicality
and objectivity out of the quantum world.

In this paper, we extend the study of the emergence of objectivity
beyond quantum theory, to arbitrary physical theories \citep{Hardy-informational-1,Barrett,General-no-broadcasting,Chiribella-purification,hardy2011,Janotta-Hinrichsen,Barnum2016}.
First, this enables us to identify which basic part of quantum mechanics
is actually responsible for objectivity, by looking at it from the
outside, in a landscape of conceivable alternative theories. Second,
this analysis can be used as a test of physical consistency of post-quantum
theories in the quest for quantum gravity \citep{modHeisenberg,Blackhole,Darkmatter},
as every quantum extension must still account for objective macroscopic
physics.

This paper is organized as follows. In section~\ref{sec:A-general-framework}
give a brief overview of the formalism to address arbitrary physical
theories. In section~\ref{sec:Classical-sub-theories} we explain
how we can identify classical sub-theories of a given physical theory
(if they exist). The notion of decoherence is introduced and examined
in section~\ref{sec:Decoherence}, while objectivity and the universal
form of objective states are studied in section~\ref{sec:Objectivity-Game}.
Finally, in section~\ref{sec:Emergence-of-composite-classical} we
identify two axioms that guarantee a local behavior in the emergence
of classicality in composite systems. Conclusions and further directions
are discussed in section~\ref{sec:Conclusions-and-discussion}.

\section{A general framework for physical theories\label{sec:A-general-framework}}

Our first challenge is to choose a suitable formalism for the study
of arbitrary physical theories. We do this by adopting the framework
of general probabilistic theories (GPTs). For more details, we refer
the reader to appendix~\ref{sec:General-probabilistic-theories}.

The \emph{state} $\rho$ of a physical system $\mathrm{A}$ is associated
with a preparation of it; after that, one can manipulate it by applying
some \emph{transformation} $\mathcal{T}$, which can possibly transform
the input system into another system $\mathrm{B}$. Finally, one can
measure the final system $\mathrm{B}$ by applying an \emph{effect}
$e$ to it: in this case, the system does not exist any more, but
is destroyed in the process. By repeating this experiment several
times, the experimenter can estimate the probability of the overall
process, denoted by $\left(e\middle|\mathcal{T}\middle|\rho\right)$.
Note that here a state is viewed as a particular kind of transformation:
a transformation without an input system. Similarly, an effect is
a transformation without an output system. In this setting, one can
set up a suitable notion of sequential and parallel composition; the
former is denoted by $\mathcal{AB}$, where $\mathcal{A}$ comes after
$\mathcal{B}$, the latter is denoted by $\mathcal{A}\otimes\mathcal{B}$.

The application of a generic non-deterministic device in an experiment
can be described as a collection of mutually exclusive processes $\left\{ \mathcal{T}_{i}\right\} _{i\in\mathsf{X}}$,
where $i\in\mathsf{X}$ represents the (classical) outcome read  by
the experimenter. We will call such a collection $\left\{ \mathcal{T}_{i}\right\} _{i\in\mathsf{X}}$
a \emph{test} (\emph{measurement} if it is a collection of effects).
If a test is deterministic (i.e.\ there is a single outcome), we
will call it a \emph{channel}.

A state is said to be \emph{pure} if the only way to write it as a
sum of other states is the trivial way: $\rho$ is pure if $\rho=\sum_{i}\rho_{i}$
implies $\rho_{i}=p_{i}\rho$, with $\left\{ p_{i}\right\} $ a probability
distribution. A non-pure state is called \emph{mixed}.

In our analysis we assume the fundamental axiom of Causality \citep{Chiribella-purification},
satisfied by both classical and quantum theory:
\begin{ax}[Causality]
The probability that a transformation occurs is independent of the
choice of tests performed on its output.
\end{ax}
Causality is equivalent to the existence of a unique deterministic
effect $u$ for every system \citep{Chiribella-purification}, which
can be used as the analog of the partial trace to discard systems
in multipartite settings.

In causal theories, we can restrict ourselves to preparations $\rho$
that are performed with certainty, (i.e.\ those for which $\left(u\middle|\rho\right)=1$)
\citep{Chiribella-purification}. These are called deterministic states,
and, given a measurement $\left\{ a_{i}\right\} _{i\in\mathsf{X}}$,
for them we have $\sum_{i\in\mathsf{X}}\left(a_{i}\middle|\rho\right)=1$.
In this situation, if all probabilities are allowed, the theory is
convex \citep{Chiribella-purification}. In particular, this means
that the state space of a theory is convex, and that all conical combinations
of valid effects that lead to a valid effect are allowed. The latter
means that effects span a convex cone.

\section{Classical sub-theories\label{sec:Classical-sub-theories}}

The states of finite-dimensional classical theory are probability
distributions over a finite set, and the effects are \emph{all} the
linear functionals that yield a number in $\left[0,1\right]$ on states.
A most notable feature of classical theory is that classical pure
states can be jointly perfectly distinguished in a single-shot measurement.

Therefore, to find classical sub-theories of a given physical theory,
we need to find pure states that are perfectly distinguishable. The
states $\left\{ \rho_{i}\right\} _{i=1}^{n}$ are said to be perfectly
distinguishable if there exists a measurement $\left\{ a_{i}\right\} _{i=1}^{n}$
such that $\left(a_{i}\middle|\rho_{j}\right)=\delta_{ij}$ for all
$i$, $j$. In addition, if there is no other state $\rho_{0}$ such
that the states $\left\{ \rho_{i}\right\} _{i=0}^{n}$ are perfectly
distinguishable, the set $\left\{ \rho_{i}\right\} _{i=1}^{n}$ is
said to be maximal. One may wonder why we are interested specifically
in perfectly distinguishable pure states as far as classical sub-theories
are concerned, rather than generic states. The reason is that it is
not restrictive to assume those states perfectly distinguishable states
to be pure. Indeed, if $\left\{ \rho_{i}\right\} _{i=1}^{n}$ are
mixed and perfectly distinguishable, then it is possible to find pure
states $\left\{ \alpha_{i}\right\} _{i=1}^{n}$ that are perfectly
distinguishable: it is enough to take $\alpha_{i}$ to be any pure
state in a convex decomposition of $\rho_{i}$ into pure states. More
details are provided in appendix~\ref{sec:Classicality}.

We are interested in the largest classical theory that can arise from
a given set of perfectly distinguishable pure states $S$, therefore
we will look for maximal sets of perfectly distinguishable pure states.
Picking one such set $S=\left\{ \alpha_{i}\right\} _{i=1}^{d}$, we
define the \emph{classical set} $\boldsymbol{\alpha}$ of dimension
$d$ as $\boldsymbol{\alpha}:=\mathrm{Conv}\left\{ \alpha_{i},i=1,\ldots,d\right\} $.
This is the simplex generated by the $\alpha_{i}$'s, and it represents
the states of a particular classical sub-theory.

At this point, we restrict the effects of the original theory to the
classical set $\boldsymbol{\alpha}$, identifying those that give
the same probabilities on all classical states. These will be the
classical effects. In appendix~\ref{subsec:Restricting-effects},
we show that, in this way, we get precisely \emph{all} effects of
the classical theory with $\boldsymbol{\alpha}$ as the state space.
In other words, with this strategy, it is enough to choose a classical
set $\boldsymbol{\alpha}$ to find a classical sub-theory of a causal
theory.

All physical theories require a classical interface by which the observer
reads the outcome of an experiment. Every fundamental theory should
be able to describe this classical interface \citep{CPT}, otherwise
we would be forced to accept an insurmountable division between the
underlying physical world, and the macroscopic one, in which the observer
performs their experiments. Consequently, a fundamental theory should
obey the following principle:
\begin{condition}[Emergence of Classical Concepts]
\label{cond:ECC}A fundamental physical theory of nature must contain
classical states (or an arbitrarily good approximation thereof).
\end{condition}
Not all theories obey condition~\ref{cond:ECC}: in figs.~\ref{fig:The-state-space}
and \ref{fig:A-section-of} we depict the state and effect spaces
of a theory that violates it.
\begin{figure}
\begin{centering}
\includegraphics[scale=0.7]{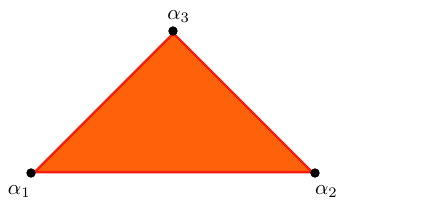}
\par\end{centering}
\caption{\label{fig:The-state-space}The state space of a restricted trit coincides
with that of a classical trit. $\alpha_{1}$, $\alpha_{2}$, $\alpha_{3}$
are the pure states.}
\end{figure}
\begin{figure}
\begin{centering}
\includegraphics[scale=0.7]{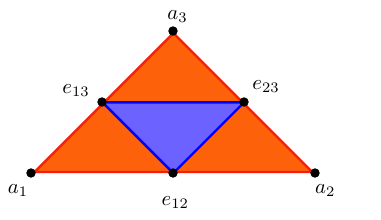}
\par\end{centering}
\caption{\label{fig:A-section-of}A cross-section of the effect convex cone
of a classical trit (in orange) and a restricted trit (in blue). $a_{1}$,
$a_{2}$, $a_{3}$ are linear functionals such that $\left(a_{i}\middle|\alpha_{j}\right)=\delta_{ij}$,
but they are not allowed effects of the restricted trit: the only
pure effects are $e_{12}$, $e_{13}$, $e_{23}$. The restriction
on effects is evident.}
\end{figure}
 The states of the basic system are the same as in the classical trit
(with pure states $\alpha_{1}$, $\alpha_{2}$, $\alpha_{3}$), but
not all linear functionals are allowed, i.e.\ there is an intrinsic
restriction on the effect space. Indeed, the only pure effects we
have are $e_{12}$, $e_{13}$, $e_{23}$, with $e_{ij}=\frac{1}{2}\left(a_{i}+a_{j}\right)$,
where $a_{i}$ is the linear functional of the trit such that $\left(a_{i}\middle|\alpha_{j}\right)=\delta_{ij}$.
This restricted trit theory has no subsets of pure states that can
be distinguished perfectly in a single-shot, and therefore no classical
states. Even more so, \emph{such absence persists in all composite
systems} (details in appendix~\ref{sec:restricted trit}). The restriction
plays a crucial role: we show that if there is no restriction on the
effects, a theory admits at least the classical bit as a sub-theory
(appendix~\ref{subsec:Classical-no-restriction}).

In the following, we will always assume that a theory satisfies condition~\ref{cond:ECC}.
Alternatively, the presence of classical states can be postulated
directly \citep{Barnum-interference,Colleagues}, or enforced by mathematical
(appendix~\ref{subsec:Classical-no-restriction}) or physical \citep{Hardy-informational-2,QPL15}
principles.

\section{Decoherence\label{sec:Decoherence}}

For a complete description of the emergence of classicality we need
to find a transition towards classical theory among the transformations
of a given theory. This provides a classical interface that emerges
from the physical description of nature \citep{CPT}.

In analogy to the well-known process of quantum decoherence \citep{Zeh-first,Zeh-second,Decoherence-review,Review-decoherence},
a similar mechanism in GPTs was studied in refs.~\citep{Selby-entanglement2,Selby-leaks,CPT,Hyperdecoherence,2roads}.
Our approach to GPT decoherence is different, as it focuses on its
\emph{minimal} properties as a process. First, note that if classicality
were reached only probabilistically, it would be an intrinsically
unstable theory, contrary to experimental evidence. This motivates
searching for decoherence among deterministic processes, i.e.\ among
the channels of the theory. Moreover, a complete decoherence should
send all states to classical states, and preserve classical states
themselves. This motivates the following definition, which characterizes
decoherence as a resource-destroying map \citep{Resource-destroying}:
\begin{defn}
\label{def:complete decoherence}Given the classical set $\boldsymbol{\alpha}$,
a channel $D_{\boldsymbol{\alpha}}$ is a \emph{complete decoherence}
if
\begin{enumerate}
\item $D_{\boldsymbol{\alpha}}\rho\in\boldsymbol{\alpha}$ for every state
$\rho$;
\item $D_{\boldsymbol{\alpha}}\gamma=\gamma$ for every $\gamma\in\boldsymbol{\alpha}$.
\end{enumerate}
\end{defn}
One can apply the complete decoherence to all effects of the theory,
which naturally produces the set of classical effects defined through
the restriction procedure introduced above (appendix~\ref{sec:Complete-decoherence}).
The question is whether, given a classical set, a complete decoherence
on it always \emph{exists}. Consider a measurement $\left\{ a_{i}\right\} _{i=1}^{d}$
that distinguishes the pure states $\left\{ \alpha_{i}\right\} _{i=1}^{d}$
perfectly. We can construct the measure-and-prepare test $\left\{ \left|\alpha_{i}\right)\left(a_{i}\right|\right\} _{i=1}^{d}$
(see definition~\ref{def:measure and prepare}). By coarse-graining
over all the outcomes of $\left\{ \left|\alpha_{i}\right)\left(a_{i}\right|\right\} _{i=1}^{d}$,
we get the channel $\widehat{D}_{\boldsymbol{\alpha}}=\sum_{i=1}^{d}\left|\alpha_{i}\right)\left(a_{i}\right|$.
\begin{prop}
\label{prop:TID}For every classical set $\boldsymbol{\alpha}$, the
channel $\widehat{D}_{\boldsymbol{\alpha}}=\sum_{i=1}^{d}\left|\alpha_{i}\right)\left(a_{i}\right|$
is a complete decoherence.
\end{prop}
The proof is in appendix~\ref{subsec:Test-induced-decoherence-and}.
In the light of this result, we will call every channel of the form
$\widehat{D}_{\boldsymbol{\alpha}}=\sum_{i=1}^{d}\left|\alpha_{i}\right)\left(a_{i}\right|$
a \emph{test-induced decoherence} (TID) with respect to the \emph{fixed}
classical set $\boldsymbol{\alpha}$.

Proposition~\ref{prop:TID} implies that in all causal theories there
always exists a complete decoherence on every classical set, induced
by measuring and forgetting the outcome. Despite this universal form,
given a classical set $\boldsymbol{\alpha}$, such a complete decoherence
can be highly non-unique (appendix~\ref{subsec:Test-induced-decoherence-and}),
a fact that was missed by previous works \citep{Selby-entanglement2,Hyperdecoherence}.
Furthermore, as opposed to quantum mechanics, there are GPTs \emph{where
mixed states can be decohered to pure ones} by the TID, as represented
in fig.~\ref{fig:square decoherence main} (see appendix~\ref{subsec:Test-induced-decoherence-and}
for details).
\begin{figure}
\begin{centering}
\includegraphics[viewport=0bp 0bp 159bp 143bp,scale=0.7]{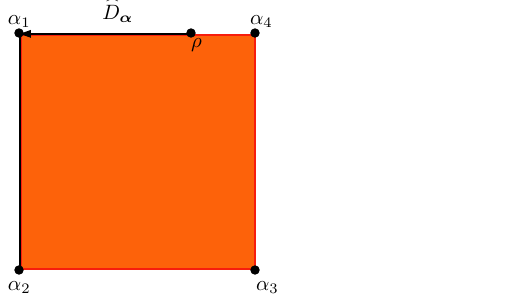}
\par\end{centering}
\caption{\label{fig:square decoherence main}In this GPT, the state space is
a square, and the vertical side in black is the classical set $\boldsymbol{\alpha}=\mathrm{Conv}\left\{ \alpha_{1},\alpha_{2}\right\} $.
The action of the test-induced decoherence $\widehat{D}_{\boldsymbol{\alpha}}$
on the \emph{mixed} state $\rho$ is represented by a black arrow.
Notice that $\rho$ is decohered to the \emph{pure} state $\alpha_{1}$;
so the test-induced decoherence in GPTs can increase the purity of
a state, \emph{contrary to what happens in quantum theory}.}
\end{figure}
 All this shows how, in general, GPTs \emph{differ from quantum} even
concerning decoherence. Nevertheless, as we show below, in general
the emergence of objectivity is something remarkably different from
decoherence.

\section{Objectivity Game\label{sec:Objectivity-Game}}

The existence of classical states and decoherence processes is still
not enough to reproduce the full classical picture, as from our everyday
experience we know that the results of measurements are objective
\citep{QDarwinism,Darwinism-branches}. To address this issue, we
use the setting of quantum Darwinism \citep{QDarwinism,Darwinism-branches},
where a system is surrounded by several fragments of environment,
each of which is accessible to one observer. In quantum theory, objective
states are SBS states \citep{Broadcasting-scenarios,Objectivity,JJ1,JJ3,JJ5,SBS-measurements,Strong-QDarwinism,korbicz2020roads},
i.e.\ states of the form $\rho=\sum_{j}p_{j}\left|j\right\rangle \left\langle j\right|_{\mathrm{S}}\otimes\rho_{j,\mathrm{E}_{1}}\otimes\ldots\otimes\rho_{j,\mathrm{E}_{n}}$,
where for every environment fragment $\mathrm{E}_{k}$ the states
$\left\{ \rho_{j,\mathrm{E}_{k}}\right\} $ have orthogonal support.

Recall that a state of a system is objective if multiple observers
can find it out without perturbing it \citep{QDarwinism,Zurek-objectivity}.
Moreover, each observer should always be able to repeat their measurement,
and always obtain the same result. To model non-disturbance, we extend
the so-called Bohr non-disturbance criterion presented in refs.~\citep{Bohr,Wiseman,Objectivity,korbicz2020roads}
to arbitrary physical theories (cf.\ also ref.~\citep{Perinotti-discord}):
\begin{defn}
\label{def:non-disturbing}A test $\left\{ \mathcal{A}_{i}\right\} _{i\in\mathsf{X}}$
is\emph{ non-disturbing} on $\rho$ if $\sum_{i\in\mathsf{X}}\mathcal{A}_{i}\rho=\rho.$
\end{defn}
We can recast the concept of objectivity as a multiplayer game, called
the \emph{objectivity game} (OG)\textcolor{blue}{, }inspired by ref.\textcolor{blue}{~\citep{Objectivity}}.
In this game, the goal is to determine the state of a target system
$\mathrm{S}$ which decoheres to a classical set $\boldsymbol{\alpha}$.
We assume that there is a special observer on system $\mathrm{S}$
acting as a referee checking the findings of $n$ players, who act
independently by testing some environment fragment $\mathrm{E}_{k}$,
correlated with the system. They win if they are able to determine
the state of the target system $\mathrm{S}$ without disturbing the
joint state $\rho_{\mathrm{S}\mathrm{E}_{1}\ldots\mathrm{E}_{n}}$.
The Bohr non-disturbance criterion is argued to be the right concept
here \citep{Objectivity}. We insist that the $n$ players should
act independently in this game, therefore we enforce the following
condition \citep{Objectivity,Strong-QDarwinism,korbicz2020roads},
which is widely accepted in the quantum research on objectivity:
\begin{condition}[Strong independence]
The only correlation between the players is the common information
about the system.
\end{condition}
In general, to determine a state, several rounds of tests are necessary,
but the players cannot change their devices between the various rounds.
For this reason, all the observers, including the referee, want to
be able to repeat their tests several times, without affecting the
outcome. This is also a necessary condition for objectivity: if something
is objective, each observer must be able to obtain the same outcome
when they probe a system. These tests are called \emph{sharply repeatable
tests} (SRTs)\textcolor{blue}{.}
\begin{defn}
$\left\{ P_{i}\right\} _{i\in\mathsf{X}}$ is a \emph{SRT} if
\[
P_{i}P_{j}=\delta_{ij}P_{i}.
\]
\end{defn}
SRTs are an operational characterization of tests that can be repeated
several times, always yielding the same outcome. This is a highly
desirable feature for a test, but do such tests exist at all? The
answer is positive in causal theories that admit perfectly distinguishable
states. Indeed, if $\left\{ \rho_{i}\right\} _{i=1}^{n}$ is a set
of perfectly distinguishable states, and $\left\{ a_{i}\right\} _{i=1}^{n}$
is the associated measurement, by Causality we can consider the measure-and-prepare
test $\left\{ \mathcal{A}_{i}\right\} _{i=1}^{n}$, with $\mathcal{A}_{i}=\left|\rho_{i}\right)\left(a_{i}\right|$.
This is an SRT because

\[
\mathcal{A}_{i}\mathcal{A}_{j}=\left|\rho_{i}\right)\left(a_{i}\middle|\rho_{j}\right)\left(a_{j}\right|=\delta_{ij}\left|\rho_{i}\right)\left(a_{i}\right|=\delta_{ij}\mathcal{A}_{i}.
\]
In quantum theory, these measure-and-prepare tests are quantum instruments
of the form $\left\{ \mathcal{M}_{j}\right\} $, where
\[
\mathcal{M}_{j}\left(\rho\right)=\mathrm{tr}\left(E_{j}\rho\right)\sigma_{j},
\]
where the $\sigma_{j}$'s have orthogonal supports, and $E_{j}$ is
the orthogonal projector onto the support of $\sigma_{j}$. In general,
however, not every SRT needs to be of this form: in quantum theory,
a von Neumann measurement with projectors of rank greater than 1 is
obviously an SRT, but it is not of the measure-and-prepare type. For
the scope of the OG, it is not important to characterize all the SRTs
of a theory; it is enough to know that they exist. Indeed, non-disturbing
SRTs were identified as providing objective information in causal
theories in ref.~\citep{Perinotti-discord}.

The first move of the OG is from the referee who performs the SRT
associated with some classical set $\boldsymbol{\alpha}$ of $\mathrm{S}$
(see appendix~\ref{sec:Sharply-repeatable-tests}). Since the outcome
is not communicated to the players, the system is decohered to $\rho_{\mathrm{S}}=\sum_{i=1}^{r}p_{i}\alpha_{i}$,
where $p_{i}>0$ for every $i$ (recall we know that such a decoherence
is always guaranteed to exist). The players win the game if they all
correctly guess the outcome of the referee. On the other hand, the
players are not restricted to this form of SRT. This is because in
the setting of quantum Darwinism, they represent fractions of the
environment, which have much more degrees of freedom than the referee's
system.

The operational meaning of objectivity is the agreement of all the
involved observers (including the referee). Therefore, an objective
state is the joint state of the referee and the players that allows
them to win the OG. To this end, we introduce SBS states for causal
theories:
\begin{defn}
An \emph{SBS} state is of the form $\rho=\sum_{i=1}^{r}p_{i}\alpha_{i}\otimes\rho_{i,\mathrm{E}_{1}}\otimes\ldots\otimes\rho_{i,\mathrm{E}_{n}}$,
where $\left\{ \alpha_{i}\right\} _{i=1}^{r}$ are perfectly distinguishable
pure states, and for every $k$, $\left\{ \rho_{i,\mathrm{E}_{k}}\right\} _{i=1}^{r}$
are perfectly distinguishable too.
\end{defn}
It is not hard to show that states in the SBS form are objective (appendix~\ref{subsec:form objective}),
so every causal theory has objective states. However, the fundamental
question is the characterization of \emph{all} objective states of
a theory. Our main result is that the states of the SBS form are the
\emph{only} objective states in every fundamental causal theory. 
\begin{thm}
\label{thm:main}In any fundamental causal theory (i.e.\ obeying
condition~\ref{cond:ECC}), the players can win the OG if and only
if the joint state is an SBS state.
\end{thm}
The proof is in appendix~\ref{subsec:form objective}. This result
is exceptionally general, for it demonstrates that only the principle
of Causality is enough to ensure the emergence of objectivity. Furthermore,
the structure of objective states is universal and isomorphic to the
quantum one.

\section{Emergence of composite classical theories\label{sec:Emergence-of-composite-classical}}

Our results obtained so far have focused only on single systems, but
here we identify the minimal assumptions to ensure classicality in
composite systems. Specifically, it is enough to impose the following
two axioms:
\begin{ax}
\label{axm:product pure}The product of two pure states is a pure
state.
\end{ax}
\begin{ax}[Information Locality \citep{Hardy-informational-2}]
\label{axm:information locality}If $\left\{ \alpha_{i}\right\} _{i=1}^{d_{\mathrm{A}}}$
is a maximal set of perfectly distinguishable pure states of $\mathrm{A}$,
and $\left\{ \beta_{j}\right\} _{j=1}^{d_{\mathrm{B}}}$ is a maximal
set of perfectly distinguishable pure states of $\mathrm{B}$, $\left\{ \alpha_{i}\otimes\beta_{j}\right\} _{i=1,}^{d_{\mathrm{A}}}\phantom{}_{j=1}^{d_{\mathrm{B}}}$
is a maximal set of perfectly distinguishable pure states of $\mathrm{AB}$.
\end{ax}
These axioms represent a ``locality constraint'' in the emergence
of classicality. Indeed if these two axioms fail, the informational
content of the classical composite system cannot be reduced to the
informational contents of each classical subsystem.

If there are no ``delocalized'' classical systems, we expect that
decohering $\mathrm{AB}$ will be the same as decohering $\mathrm{A}$
and $\mathrm{B}$ separately, i.e.\ $D_{\boldsymbol{\alpha\beta}}=D_{\boldsymbol{\alpha}}\otimes D_{\boldsymbol{\beta}}$.
Even if the theory satisfies both axioms~\ref{axm:product pure}
and \ref{axm:information locality}, this property may not be satisfied
by general complete decoherences. However, it is so for TIDs (appendix~\ref{sec:Emergence-of-composite}).

\section{Conclusions and discussion\label{sec:Conclusions-and-discussion}}

In summary, Causality alone, in conjunction with the principle of
\emph{Emergence of Classical Concepts} (condition~\ref{cond:ECC}),
is the backbone of our results. Here we solve the problem of the emergence
of objectivity beyond quantum, identifying SBS and the process leading
to it as its ultimate origin. Our analysis has striking outcomes:
unlike in quantum mechanics \citep{Zurek-objectivity,QDarwinism},
objectivity and decoherence are two distinct phenomena. Indeed, decoherence
behavior can be exceedingly different from the quantum case, while
objectivity and SBS are universal across causal theories. For instance,
in appendix~\ref{subsec:Test-induced-decoherence-and}, we show that
complete decoherences can even increase the purity of a state in some
GPTs; or in another example, there is an uncountable number of distinct
complete decoherences. This also demonstrates how GPTs can differ
radically from quantum theory, even within the scope of our analysis.
In light of this, our result about the universality of the SBS form
is even more surprising: it shows that the emergence of objectivity
is, instead, a transversal phenomenon in physics that unifies physical
theories with quite different behaviors, as illustrated in fig.~\ref{fig:summary-1}.
\begin{figure}
\begin{centering}
\includegraphics[width=1\columnwidth]{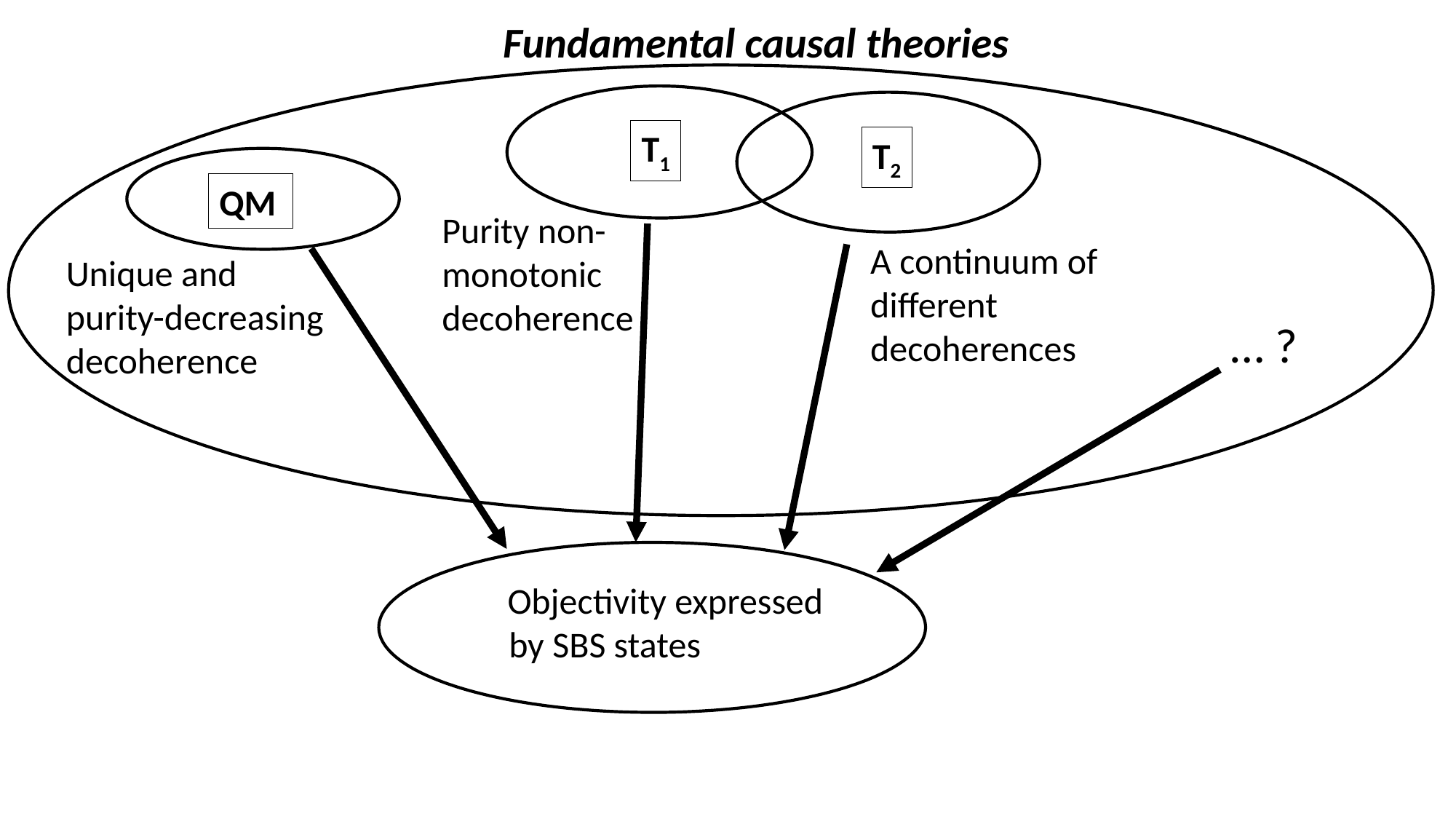}
\par\end{centering}
\caption{\label{fig:summary-1}Different fundamental causal theories have different
and unusual decoherence behaviors. Unlike in quantum mechanics (QM),
in some theories the decoherence can increase the purity of a state
(see example~\ref{exa:purity-decoherence}). In other situations,
given one set of classical states, there is a \emph{continuum of different
decoherences} on it (see example~\ref{exa:uniqueness}), which contrasts
sharply with quantum theory, where the decoherence is unique. Despite
this, objectivity is universal, and the form of objective states is
the same (SBS) across all theories.}
\end{figure}

However, we note that Causality is not enough to guarantee the emergence
of \emph{localized} classical composite systems: it is not always
possible to reduce the information of the decohered system to the
information of its sub-systems. We showed two additional axioms that
are sufficient to ensure an appropriate non-holistic classical behavior.

Our results support the approach to objectivity presented in refs.~\citep{J1PRL,Objectivity,JJ1,JJ3,Miro,JJ4,JJ5,SBS-measurements,korbicz2020roads}
based on SBS states (recently shown to be stronger than the notion
of quantum Darwinism \citep{Strong-QDarwinism,Strong-QDarwinism2,korbicz2020roads}).
In particular, they identify the validity of the SBS approach far
beyond the limits of quantum mechanics. Our findings also suggest
that other approaches to classicality, such as quantum Darwinism and
the associated broadcasting of information to the environment \citep{Selby-leaks,Muller-Darwinism}
could be extended to GPTs, opening a fruitful research field we intend
to investigate in subsequent works. On the other hand, the universality
of the shape of objective states suggests that Causality is really
a strong assumption and, to find some unique behaviors, one should
weaken it, e.g.\ by bringing in the structure of relativistic space-time
\citep{Hardy-quantum-gravity,Hardy-causal,Spacetime}. Last but not
least, concerning the increase of purity by decoherence in some non-quantum
GPTs, one may have an interesting alternative: either we have some
paradoxical process where decoherence increases the information about
a system (if we follow the standard quantum intuition), or one must
reconsider the concept of information of the system itself, at least
partially decoupling it from the concept of purity \citep{Entropy-Barnum,Entropy-Short,Entropy-Kimura}.
We leave this as an open issue for further research.
\begin{acknowledgments}
The work was made possible through the support of grants from the
John Templeton Foundation. The opinions expressed in this publication
are those of the authors and do not necessarily reflect the views
of the John Templeton Foundation. C.M.S.\ acknowledges support from
the Engineering and Physical Sciences Research Council (EPSRC) through
the doctoral training Grant No.\ 1652538, from Oxford-Google DeepMind
graduate scholarship, from the Natural Sciences and Engineering Research
Council of Canada (NSERC), from the Pacific Institute for the Mathematical
Sciences (PIMS), and from a Faculty of Science, University of Calgary,
Grand Challenge award. C.M.S.\ acknowledges the hospitality of the
National Quantum Information Centre in Sopot, and useful comments
by Jon Barrett. R.S.\ acknowledges the support of Comision Nacional
de Investigacion Ciencia y Tecnologia (CONICYT) Programa de Formacion
Capital Humano Avanzado/Beca de Postdoctorado en el extranjero (BECAS
CHILE) 74160002, John Templeton Foundation, grant Sonata Bis 5 (Grant
No.\ 2015/18/E/ST2/00327) from the National Science Centre, and the
Foundation for Polish Science through the grant TEAM-NET project (Contract
No.\ POIR.04.04.00-00- 17C1/18-00). J.K.K.\ and P.H.\ acknowledge
the support of the John Templeton Foundation. P.H.\ acknowledges
partial support from the Foundation for Polish Science (IRAP project,
ICTQT, Contract No.\ MAB/2018/5, co-financed by EU within Smart Growth
Operational Programme).
\end{acknowledgments}

\bibliographystyle{apsrev4-2}
\bibliography{Objectivity_main}

\appendix
\onecolumngrid

\section{General probabilistic theories\label{sec:General-probabilistic-theories}}

GPTs is a framework for a theory-independent description of physical
probabilistic processes. They have been used in several successful
reconstructions of quantum theory from information-theoretic postulates
\citep{Hardy-informational-1,Chiribella-informational,Hardy-informational-2,Brukner,masanes,Barnum-interference,Diagrammatic},
and they are the subject of active research in the quantum information
community \citep{Chiribella-Scandolo-entanglement,TowardsThermo,Purity,Colleagues,Ciaran-1,Ciaran-2,Selby-1,Selby-2,Takagi-Regula}.
The essence of this approach is that any physical theory must describe
experiments performed in a laboratory, and predict the probabilities
of their outcomes. These experiments are usually carried out by connecting
several devices. Each device represents a physical process, and wires
connecting them carry physical systems. Therefore for every physical
transformation $\mathcal{T}$ transforming system $\mathrm{A}$ into
system $\mathrm{B}$ (e.g.\ a beam splitter, a Stern-Gerlach magnet,
etc.), it is natural to represent it as\[ \begin{aligned}\Qcircuit @C=1em @R=.7em @!R { &\qw \poloFantasmaCn{\rA} &\gate{\cT} &\qw \poloFantasmaCn{\rB} &\qw }\end{aligned}~. \]Some
devices have no input, others have no output. They are represented
respectively as\[ \begin{aligned}\Qcircuit @C=1em @R=.7em @!R { &\prepareC{\rho} &\qw \poloFantasmaCn{\rA} &\qw }\end{aligned}~, \]and
\[ \begin{aligned}\Qcircuit @C=1em @R=.7em @!R {&\qw \poloFantasmaCn{\rA} &\measureD{e}}\end{aligned}~. \]Processes
with no input are called \emph{states}, and those with no output are
called \emph{effects}. If the output system of a transformation $\mathcal{A}$
matches the input system of another transformation $\mathcal{B}$,
we can apply $\mathcal{B}$ after $\mathcal{A}$, and get a new transformation,
denoted $\mathcal{BA}$. This corresponds to connecting the two associated
devices in sequence:\[ \begin{aligned}\Qcircuit @C=1em @R=.7em @!R { & \qw \poloFantasmaCn{\rA} &\gate{\cA} &\qw \poloFantasmaCn{\rB} &\gate{\cB}&\qw \poloFantasmaCn{\rC} &\qw}\end{aligned}~. \]Similarly,
two transformations $\mathcal{A}$ and $\mathcal{B}$ can be applied
independently, at the same time, on different systems. In this case,
the resulting transformation is denoted by $\mathcal{A}\otimes\mathcal{B}$,
and the two associated devices are composed in parallel:\[ \begin{aligned}\Qcircuit @C=1em @R=.7em @!R { & \qw \poloFantasmaCn{\rA} &\gate{\cA} &\qw \poloFantasmaCn{\rB} &\qw \\  & \qw \poloFantasmaCn{\rC} &\gate{\cB}&\qw \poloFantasmaCn{\rD} &\qw}\end{aligned}~. \]One
can build arbitrary circuits by connecting these devices, such as\[
\begin{aligned}\Qcircuit @C=1em @R=.7em @!R {&&\qw&\qw&\multigate{1}{} &\gate{} &\multimeasureD{1}{} \\& & \gate{}   & \multigate{1}{} &\ghost{} &\qw &\ghost{}\\ &\prepareC{} &\qw &\ghost{} &\qw &\gate{}&\qw}\end{aligned}~.
\]This can be read at the same time as an instruction about how to build
an actual experiment, and as the way physical processes are connected
in the same experiment. This framework allows one to treat states,
effects, and transformations on equal footing, by introducing a special
system, trivial system $\mathrm{I}$, which represents the lack of
a system. For this reason, the composition of system $\mathrm{A}$
with the trivial system does not involve any change: $\mathrm{AI}=\mathrm{IA}=\mathrm{A}$.

In general, the experimenter does not have full control over the transformation
they can implement; this is because in nature there are also non-deterministic
processes. Therefore, what we can say is that every device in an experiment
implements a collection of mutually exclusive alternatives. Only one
of them can occur in a run of the experiment, and the experimenter
can read which process actually occurred by looking at the outcome
of the experiment. For this reason, we can associate a collection
of transformations $\left\{ \mathcal{T}_{i}\right\} _{i\in\mathsf{X}}$,
called \emph{test}, with every device, where $i$ is the outcome,
and $\mathsf{X}$ the set of outcomes. A special kind of test are
\emph{measurements} (or observation-tests) $\left\{ a_{i}\right\} _{i\in\mathsf{X}}$,
which are collections of effects. It is therefore natural to ask ourselves
about the probability that a particular transformation occurs in an
experiment. Probabilities are represented by circuits with no external
wires, such as\[\begin{aligned}\Qcircuit @C=1em @R=.7em @!R { & \multiprepareC{1}{\rho_i} & \qw \poloFantasmaCn{\rA} & \gate{\cA_j} & \qw \poloFantasmaCn{\rA'} & \gate{\cA'_k} & \qw \poloFantasmaCn{\rA''} &\measureD{a_l} \\ & \pureghost{\rho_i} & \qw \poloFantasmaCn{\rB} & \gate{\cB_m} & \qw \poloFantasmaCn{\rB'} &\qw &\qw &\measureD{b_n} }\end{aligned}~. \]This
circuit represents the joint probability $p_{ijklmn}$ to observe
all these specific transformations in the experiment. If a test is
deterministic, i.e. there is only one transformation associated with
it, it is called \emph{channel}.

We will often make use of the following short-hand notations, inspired
by quantum theory, to mean some common diagrams occurring in our analysis.
\begin{enumerate}
\item \[ \left(a\middle|\rho\right)~:=\!\!\!\!\begin{aligned}\Qcircuit @C=1em @R=.7em @!R { & \prepareC{\rho} & \qw \poloFantasmaCn{\rA} &\measureD{a}}\end{aligned}~; \]
\item \[ \left(a\middle|\cC\middle|\rho\right)~:=\!\!\!\!\begin{aligned}\Qcircuit @C=1em @R=.7em @!R { & \prepareC{\rho} & \qw \poloFantasmaCn{\rA} &\gate{\cC} &\qw \poloFantasmaCn{\rB} &\measureD{a}}\end{aligned}~; \]
\item \label{enu:measure and prepare}\[ \left|\rho\right)\left(a\right|~:=~\begin{aligned}\Qcircuit @C=1em @R=.7em @!R { & \qw \poloFantasmaCn{\rA} &\measureD{a}&\prepareC{\rho}&\qw \poloFantasmaCn{\rB} &\qw}\end{aligned}~. \]
\end{enumerate}
In particular, the transformation represented in~\ref{enu:measure and prepare}
is called a \emph{measure-and-prepare} transformation, because first
the effect $a$ (representing a measurement outcome) occurs, and then
the state $\rho$ is prepared.

Now, for any state $\rho$ and every effect $a$, $\left(a\middle|\rho\right)\in\left[0,1\right]$,
whereby a state of system $\mathrm{A}$ becomes a map from the set
of effects $\mathsf{Eff}\left(\mathrm{A}\right)$ of $\mathrm{A}$
to the unit interval $\left[0,1\right]$: $\rho:\mathsf{Eff}\left(\mathrm{A}\right)\rightarrow\left[0,1\right]$.
This leads naturally to the following definition
\begin{defn}
\label{def:tomographically-distinct}Two states $\rho$ and $\sigma$
on the same system are \emph{tomographically distinct} if there exists
an effect $a$ such that $\left(a\middle|\rho\right)\neq\left(a\middle|\sigma\right)$.
\end{defn}
The idea behind this definition is that two states (i.e.\ the associated
preparations of a system) are indistinguishable if there is no measurement
to witness their difference. In this case they must be identified
as states.

Similarly, every effect gives rise to a map from the set of states
$\mathsf{St}\left(\mathrm{A}\right)$ to the unit interval, and one
identifies effects that produce the same probabilities on all states.

As states and effects are maps to a subset of real numbers, they can
be summed, and one can take their multiple by a real number. In this
way, the set of states and the set of effects become spanning sets
of \emph{real} vector spaces, denoted $\mathsf{St}_{\mathbb{R}}\left(\mathrm{A}\right)$
and $\mathsf{Eff}_{\mathbb{R}}\left(\mathrm{A}\right)$, which we
assume to be finite-dimensional. If one considers only linear combinations
with non-negative coefficients (conical combinations), one obtains
the \emph{cone} of states $\mathsf{St}_{+}\left(\mathrm{A}\right)$
and the cone of effects $\mathsf{Eff}_{+}\left(\mathrm{A}\right)$.
In this way, one can see how the convex geometric approach to GPTs
arises \citep{Barrett,Barnum2016}. Once the cone of states is defined,
one can consider the \emph{dual cone} $\mathsf{St}_{+}^{*}\left(\mathrm{A}\right)$:
this is the cone of linear functionals that are non-negative on $\mathsf{St}_{+}\left(\mathrm{A}\right)$.
Clearly $\mathsf{Eff}_{+}\left(\mathrm{A}\right)\subseteq\mathsf{St}_{+}^{*}\left(\mathrm{A}\right)$,
but we will discuss this inclusion in greater detail in appendix~\ref{sec:Classicality},
to examine its consequences for the emergence of classicality.

In this setting, a transformation from $\mathrm{A}$ to $\mathrm{B}$
is a completely positive map from $\mathsf{St}_{\mathbb{R}}\left(\mathrm{A}\right)$
to $\mathsf{St}_{\mathbb{R}}\left(\mathrm{B}\right)$. Here a \emph{positive}
map is a map sending an element of the input cone of states to an
element of the output cone of states. A map $\mathcal{T}$ is \emph{completely
positive} if $\mathcal{T}\otimes\mathcal{I}_{\mathrm{S}}$ is a positive
map, for any system $\mathrm{S}$, where $\mathcal{I}_{\mathrm{S}}$
is the identity on system $\mathrm{S}$. Complete positivity plays
a crucial role in defining tomographically distinct transformations
\citep{Chiribella-purification,QuantumFromPrinciples}.
\begin{defn}
Two transformations $\mathcal{A}$ and $\mathcal{B}$ from system
$\mathrm{A}$ to system $\mathrm{B}$ are \emph{tomographically distinct}
if there exists a system $\mathrm{S}$ and a bipartite state $\rho\in\mathsf{St}\left(\mathrm{AS}\right)$
such that\[ \begin{aligned}\Qcircuit @C=1em @R=.7em @!R { & \multiprepareC{1}{\rho} & \qw \poloFantasmaCn{\rA} & \gate{\cA} & \qw \poloFantasmaCn{\rB} & \qw \\ & \pureghost{\rho} & \qw \poloFantasmaCn{\rS} & \qw &\qw &\qw }\end{aligned}~\neq\!\!\!\!\begin{aligned}\Qcircuit @C=1em @R=.7em @!R { & \multiprepareC{1}{\rho} & \qw \poloFantasmaCn{\rA} & \gate{\cB} & \qw \poloFantasmaCn{\rB} & \qw \\ & \pureghost{\rho} & \qw \poloFantasmaCn{\rS} & \qw &\qw &\qw }\end{aligned}~. \]
\end{defn}
If the theory satisfies an axiom called Local Tomography \citep{Araki-local-tomography,Bergia-local-tomography,Hardy-informational-1,Chiribella-purification,masanes,Hardy-informational-2},
by which product effects are enough to do tomography on bipartite
states, as in quantum theory, then the ancillary system $\mathrm{S}$
is not necessary to distinguish transformations \citep{Chiribella-purification}.

The linear structure introduced above allows us to talk about the
coarse-graining of tests. Suppose one has the test $\left\{ \mathcal{A}_{i}\right\} _{i\in\mathsf{X}}$.
The action of coarse-graining means joining together some of the outcomes
of this test to build a different test. This concept is easily explained
by fig.~\ref{fig:coarse-graining}.
\begin{figure}
\begin{centering}
\includegraphics[scale=0.7]{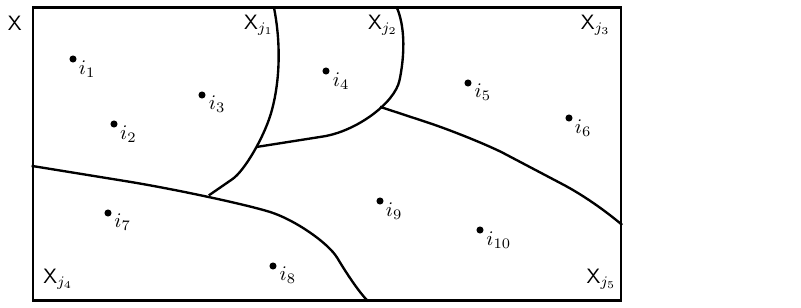}
\par\end{centering}
\caption{\label{fig:coarse-graining}The outcome set $\mathsf{X}$ of the test
$\left\{ \mathcal{A}_{i}\right\} _{i\in\mathsf{X}}$ has 10 outcomes.
To perform a coarse-graining of it, we lump together some of its outcomes,
relabeling them as a new outcome. For example, the outcomes $i_{1}$,
$i_{2}$, and $i_{3}$ are relabeled as $j_{1}$. This gives rise
to a partition $\left\{ \mathsf{X}_{j}\right\} _{j\in\mathsf{Y}}$
of $\mathsf{X}$. We associate a new transformation with each set
in the partition, such that it is the sum of the transformations associated
with the outcomes contained in that set. Thus $\mathcal{B}_{j_{1}}=\mathcal{A}_{i_{1}}+\mathcal{A}_{i_{2}}+\mathcal{A}_{i_{3}}$.
The new test $\left\{ \mathcal{B}_{j}\right\} _{j\in\mathsf{Y}}$
has 5 outcomes.}
\end{figure}
Formally, a test $\left\{ \mathcal{B}_{j}\right\} _{j\in\mathsf{Y}}$
is a \emph{coarse-graining} of the test $\left\{ \mathcal{A}_{i}\right\} _{i\in\mathsf{X}}$
if there exists a partition $\left\{ \mathsf{X}_{j}\right\} _{j\in\mathsf{Y}}$
of $\mathsf{X}$ such that $\mathcal{B}_{j}=\sum_{i\in\mathsf{X}_{j}}\mathcal{A}_{i}$.
In this case we say that $\left\{ \mathcal{A}_{i}\right\} _{i\in\mathsf{X}}$
is a \emph{refinement} of $\left\{ \mathcal{B}_{j}\right\} _{j\in\mathsf{Y}}$.
We also say that every transformation $\mathcal{A}_{i}$, with $i\in\mathsf{X}_{j}$,
is a refinement of the transformation $\mathcal{B}_{j}$. Clearly,
by performing the coarse-graining over all the outcomes of a test,
we obtain a deterministic test, i.e.\ a channel: $\mathcal{A}=\sum_{i\in\mathsf{X}}\mathcal{A}_{i}$.

The natural question at this stage is to understand when a transformation
$\mathcal{T}$ is of ``primitive nature'', or instead arises from
the coarse-graining of other transformations in some experiment.
\begin{defn}
A transformation $\mathcal{T}$ is \emph{pure} if \emph{all} its refinements
are of the form $p_{i}\mathcal{T}$, where $\left\{ p_{i}\right\} $
is a probability distribution. A non-pure transformation is called
\emph{mixed}.
\end{defn}

\subsection{Causality and its consequences\label{subsec:Consequences-of-Causality}}

In our research, the main requirement we impose on a physical theory
is that the propagation of information follows a ``temporal order'',
therefore the result of a process can influence a future process,
but never a process in the past. Causal theories \citep{Chiribella-purification}
are those that satisfy this requirement, which is expressed precisely
by the following axiom:
\begin{ax}[Causality \citep{Chiribella-purification}]
For every state $\rho$, take two measurements $\left\{ a_{i}\right\} _{i\in\mathsf{X}}$
and $\left\{ b_{j}\right\} _{j\in\mathsf{Y}}$. One has
\[
\sum_{i\in\mathsf{X}}\left(a_{i}\middle|\rho\right)=\sum_{j\in\mathsf{Y}}\left(b_{j}\middle|\rho\right).
\]
\end{ax}
Causality is also equivalent to the existence of a unique deterministic
effect $u$ \citep{Chiribella-purification}. We can use this deterministic
effect to discard systems when dealing with composite systems. The
marginal of a bipartite state can be defined as:
\[
\rho_{\mathrm{A}}=\mathrm{tr}_{\mathrm{B}}\rho_{\mathrm{AB}}:=\left(\mathcal{I}_{\mathrm{A}}\otimes u_{\mathrm{B}}\right)\rho_{\mathrm{AB}},
\]
where $\mathcal{I}_{\mathrm{A}}$ denotes the identity channel on
system $\mathrm{A}$. Sometimes we will keep $\mathrm{tr}$ as a notation
for the unique deterministic effect when it is applied directly to
states.

In causal theories, the set of states $\mathsf{St}\left(\mathrm{A}\right)$
of a system $\mathrm{A}$ has a particular structure: it can be divided
in two disjoint subsets, the set for which $\left(u\middle|\rho\right)=1$,
and the set for which $\left(u\middle|\rho\right)<1$. The former
is called the set of \emph{normalized} states, and corresponds to
states that can be prepared deterministically. States of the latter
form are called sub-normalized states, which are equivalent to a state
that can be prepared only probabilistically, where $\left(u\middle|\rho\right)$
gives exactly that probability. In causal theories, we can always
write a sub-normalized state as a probabilistic rescaling of a normalized
state, therefore it is enough for our analysis to consider only state
spaces of normalized states.

This gives rise naturally to the notion of a cone of states, denoted
by $\mathsf{St}_{+}\left(\mathrm{A}\right)$. A geometric picture
of this is given in fig.~\ref{fig:causal cone}.
\begin{figure}
\begin{centering}
\includegraphics[viewport=0bp 0bp 342bp 139bp]{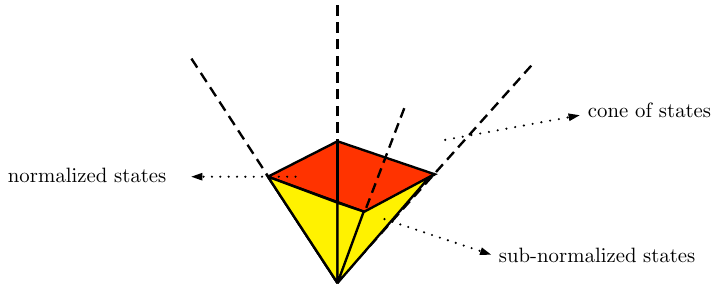}
\par\end{centering}
\caption{\label{fig:causal cone}The cone of states of a causal theory. The
set of normalized states (in orange) is given by the intersection
of the hyperplane defined by $\left(u\middle|\rho\right)=1$ with
the cone of states $\mathsf{St}_{+}\left(\mathrm{A}\right)$ . Below
that hyperplane are the sub-normalized states (in yellow). The colored
part in the cone of states is the set of states $\mathsf{St}\left(\mathrm{A}\right)$.
The white part corresponds to super-normalized elements of $\mathsf{St}_{+}\left(\mathrm{A}\right)$,
which are non-physical.}
\end{figure}

It is possible to show that in a causal theory where probabilities
in the whole range $\left[0,1\right]$ are allowed, the state space
is a \emph{compact convex set} \citep{Chiribella-purification,QuantumFromPrinciples,Diagrammatic}.
Note that given a measurement $\left\{ a_{i}\right\} _{i\in\mathsf{X}}$,
we have $\sum_{i\in\mathsf{X}}a_{i}=u$, because the sum means a coarse-graining
over all the effects of the measurement, and therefore it must be
the unique deterministic effect. Hence when we apply that measurement
on a state $\rho$, it yields a probability distribution $\left\{ p_{i}\right\} _{i\in\mathsf{X}}$,
where $p_{i}:=\left(a_{i}\middle|\rho\right)$. This shows that every
state is a probabilistic assignment to any measurement, therefore
recovering the usual picture of states in the convex approach to GPTs
\citep{Barrett,General-no-broadcasting,Barnum2016}. Similarly, it
is easy to prove that channels preserve the deterministic effect:
$u\mathcal{C}=u$. This is the generalization of the quantum property
that channels are trace-preserving. In particular, if we have a test
$\left\{ \mathcal{A}_{i}\right\} _{i\in\mathsf{X}}$, one has $\sum_{i\in\mathsf{X}}u\mathcal{A}_{i}=u$
\citep{Chiribella-purification}.

If in a theory information flows from the past to the future, like
in causal theories, it is possible to choose what experiment to perform
now based on the outcome of a previous one \citep{Chiribella-purification,QuantumFromPrinciples}.
This fact is so strongly linked to Causality that the ability to perform
all classically-controlled experiments (i.e.\ chosen according to
the outcome of a previous experiment) is equivalent to Causality itself
\citep{QuantumFromPrinciples,Diagrammatic}. A particular example
of a classically-controlled test is a measure-and-prepare test.
\begin{defn}
\label{def:measure and prepare}A test $\left\{ \mathcal{A}_{i}\right\} _{i\in\mathsf{X}}$
is \emph{measure-and-prepare} if $\mathcal{A}_{i}=\left|\rho_{i}\right)\left(a_{i}\right|$
for some measurement $\left\{ a_{i}\right\} _{i\in\mathsf{X}}$. The
channel $\mathcal{A}=\sum_{i\in\mathsf{X}}\left|\rho_{i}\right)\left(a_{i}\right|$
is said to be \emph{measure-and-prepare} as well.
\end{defn}
A measure-and-prepare test is a classically-controlled test, because
the state $\rho_{i}$ is prepared if the effect $a_{i}$ happens before.
Note that the channel $\mathcal{A}$ is the complete coarse-graining
over the outcomes of the measure-and-prepare test $\left\{ \mathcal{A}_{i}\right\} _{i\in\mathsf{X}}$.

\section{Classicality\label{sec:Classicality}}

In this appendix we collect interesting and useful facts about classicality
in GPTs. We start from some of the fundamental properties of classical
states. Geometrically, the state space is a simplex,
\begin{figure}
\begin{centering}
\includegraphics[scale=0.7]{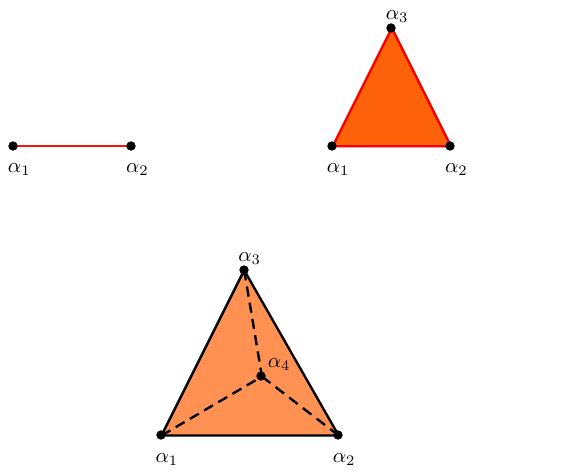}
\par\end{centering}
\caption{\label{fig:simplices}Simplices up to dimension 3. The vertices are
the pure states, namely point-like probability distributions $\alpha_{i}$.
Every other state can be obtained as a convex combination of the vertices.}
\end{figure}
 with all point-like probability distributions ($\left(\begin{array}{cccc}
1 & 0 & \ldots & 0\end{array}\right)^{T}$ and permutations) as vertices (fig.~\ref{fig:simplices}), and the
unique deterministic effect is the row-vector $u=\left(\begin{array}{ccc}
1 & \ldots & 1\end{array}\right)$. Now, we explain how classical theory can be singled out among all
other causal theories. It turns out that its key feature is that all
pure states are jointly perfectly distinguishable. This means that
there exists a measurement that distinguishes them perfectly in a
single shot\textcolor{blue}{, }as explained in section~\textcolor{blue}{\ref{sec:Classical-sub-theories}.}
\begin{prop}
If \emph{all} pure states of a causal theory are perfectly distinguishable,
the theory is classical.
\end{prop}
\begin{proof}
Suppose we have $n$ pure states $\left\{ \psi_{i}\right\} _{i=1}^{n}$
of some system $\mathrm{A}$. Requiring that they be perfectly distinguishable
implies that they are linearly independent as vectors in $\mathsf{St}_{\mathbb{R}}\left(\mathrm{A}\right)$.
To see it, let $\left\{ a_{i}\right\} _{i=1}^{n}$ be the associated
measurement. Then if we consider $\sum_{i=1}^{n}\lambda_{i}\psi_{i}=0$,
where $\lambda_{i}\in\mathbb{R}$, and we apply $a_{j}$ to both sides,
we get $\lambda_{j}=0$, for all $j=1,\ldots,n$. Being $n$ linearly
independent vectors, they span an $n$-dimensional vector space. This
means that the state space is a simplex. Now, let us examine the set
of effects. By a similar argument, the effects $\left\{ a_{i}\right\} _{i=1}^{n}$
are linearly independent. Let us also show that these effects are
pure. Suppose by contradiction that a generic $a_{i}$ is \emph{not}
pure, i.e.\ $a_{i}=\sum_{k}e_{k,i}$. Taking $j\neq i$, the fact
that $\left(a_{i}\middle|\psi_{j}\right)=0$ implies that $\left(e_{k,i}\middle|\psi_{j}\right)=0$.
Since effects are completely defined by their action on a basis of
the vector space of states, $e_{k,i}=\lambda_{k,i}a_{i}$, where $\lambda_{k,i}$
is a non-negative number. The condition $a_{i}=\sum_{k}e_{k,i}$ for
every $i$ implies that $\left\{ \lambda_{k,i}\right\} $ is a probability
distribution. This shows that every $a_{i}$ is a pure effect. Hence
the cone spanned by the $\left\{ a_{i}\right\} _{i=1}^{n}$, which
coincides with the dual cone $\mathsf{St}_{+}^{*}\left(\mathrm{A}\right)$,
is the whole effect cone $\mathsf{Eff}_{+}\left(\mathrm{A}\right)$.
In this case all allowed mathematical effects are physical too, so
there is no restriction on the effects. This is therefore classical
theory.
\end{proof}
This motivates the choice of classical states as the convex hull of
a maximal set of perfectly distinguishable pure states. Considering
only pure states is not restrictive. Indeed, suppose a causal theory
has $n$ perfectly distinguishable mixed states $\left\{ \rho_{i}\right\} _{i=1}^{n}$,
which are distinguished perfectly by the measurement $\left\{ a_{i}\right\} _{i=1}^{n}$.
Then, for every mixed state $\rho_{i}$, we can find a pure state
$\alpha_{i}$ such that $\rho_{i}=p_{i}\alpha_{i}+\left(1-p_{i}\right)\sigma_{i}$,
where $p_{i}\in\left(0,1\right)$, and $\sigma_{i}$ is another state,
for every $i=1,\ldots,n$. Then, for $j\neq i$, the fact that $\left(a_{j}\middle|\rho_{i}\right)=0$
implies that $p_{i}\left(a_{j}\middle|\alpha_{i}\right)+\left(1-p_{i}\right)\left(a_{j}\middle|\sigma_{i}\right)=0$.
Since all terms are non-negative, the only possibility is that $\left(a_{j}\middle|\alpha_{i}\right)=0$
(and $\left(a_{j}\middle|\sigma_{i}\right)=0$). Additionally, since
$\left(a_{i}\middle|\rho_{i}\right)=1$, we have $p_{i}\left(a_{i}\middle|\alpha_{i}\right)+\left(1-p_{i}\right)\left(a_{i}\middle|\sigma_{i}\right)=1$.
This convex combination of non-negative real numbers less than or
equal to 1 can attain its maximum 1 if and only if $\left(a_{i}\middle|\alpha_{i}\right)=1$
(and $\left(a_{i}\middle|\sigma_{i}\right)=1$). Therefore, the pure
states$\left\{ \alpha_{i}\right\} _{i=1}^{n}$ are perfectly distinguished
by the same measurement $\left\{ a_{i}\right\} _{i=1}^{n}$ that distinguishes
the mixed states $\left\{ \rho_{i}\right\} _{i=1}^{n}$, as $\left(a_{j}\middle|\alpha_{i}\right)=\delta_{ij}$.
This shows us that in any physical theory with perfectly distinguishable
states, it is always possible to pick perfectly distinguishable pure
states to construct classical sub-theories.\textcolor{blue}{}

\subsection{Effects for a classical sub-theory\label{subsec:Restricting-effects}}

Once we pick a classical set $\boldsymbol{\alpha}$, which is the
state space of a classical sub-theory of a given theory, we must find
what effects to consider. Indeed, the counterexample of the restricted
trit in appendix~\ref{sec:restricted trit} will show that the choice
of effects can have dramatic consequences for the structure of a theory.
Even if the state space looks classical, as for the restricted trit,
the theory can be very different from classical theory.

Given a classical set $\boldsymbol{\alpha}$, the natural way to assign
effects to this classical sub-theory is to restrict the effects of
the original theory to the set $\boldsymbol{\alpha}$, identifying
those that are not tomographically distinct on $\boldsymbol{\alpha}$.
More precisely, let us introduce the following equivalence relation
on the original set of effects $\mathsf{Eff}\left(\mathrm{A}\right)$:
$e\sim_{\boldsymbol{\alpha}}f$ if $\left(e\middle|\gamma\right)=\left(f\middle|\gamma\right)$
for every classical state $\gamma$ in $\boldsymbol{\alpha}$. The
set of effects of the classical sub-theory is the set of equivalence
classes $\mathsf{Eff}\left(\mathrm{A}\right)/\boldsymbol{\alpha}:=\mathsf{Eff}\left(\mathrm{A}\right)/\sim_{\boldsymbol{\alpha}}$. 

We need to show that this restricted set of effects $\mathsf{Eff}\left(\mathrm{A}\right)/\boldsymbol{\alpha}$
is actually the set of effects of some classical theory. Recall that
in classical theory, every element in the cone of effects arises as
a conical combination of the effects that distinguish the pure states
perfectly. In our setting, this means checking that every element
of $\mathsf{Eff}_{+}\left(\mathrm{A}\right)/\boldsymbol{\alpha}$
arises as a conical combination of the equivalence classes $\left[a_{i}\right]$
of the effects that distinguish the pure states $\alpha_{i}$ in $\boldsymbol{\alpha}$.
Note that it is not hard to see that $\mathsf{Eff}_{+}\left(\mathrm{A}\right)/\boldsymbol{\alpha}$
is still a cone, with the sum and the multiplication by a scalar inherited
from $\mathsf{Eff}_{+}\left(\mathrm{A}\right)$. Consider a generic
element $\xi$ in $\mathsf{Eff}_{+}\left(\mathrm{A}\right)$, and
let us show that it is in the same equivalence class as $\xi'=\sum_{i=1}^{d}\lambda_{i}a_{i}$,
where $\lambda_{i}=\left(\xi\middle|\alpha_{i}\right)$, for all $i$.
By linearity, to check the equivalence of two elements of $\mathsf{Eff}_{+}\left(\mathrm{A}\right)$,
it is enough to check that they produce the same numbers when applied
to all pure states $\alpha_{j}$. Now,
\[
\left(\xi'\middle|\alpha_{j}\right)=\sum_{i=1}^{d}\lambda_{i}\left(a_{i}\middle|\alpha_{j}\right)=\lambda_{j}=\left(\xi\middle|\alpha_{j}\right).
\]
This shows that the effect cone $\mathsf{Eff}_{+}\left(\mathrm{A}\right)/\boldsymbol{\alpha}$
of the sub-theory is actually a classical effect cone, generated by
the effects that perfectly distinguish the pure states in $\boldsymbol{\alpha}$.

\subsection{Classical sub-theories from the no-restriction hypothesis\label{subsec:Classical-no-restriction}}

In appendix~\ref{sec:General-probabilistic-theories}, we saw how
one can define the cone of states $\mathsf{St}_{+}\left(\mathrm{A}\right)$
and its dual $\mathsf{St}_{+}^{*}\left(\mathrm{A}\right)$, given
by linear functionals that are non-negative on $\mathsf{St}_{+}\left(\mathrm{A}\right)$.
We also saw that one can define the cone of effects $\mathsf{Eff}_{+}\left(\mathrm{A}\right)$,
generated by conical combinations of effects. Clearly, all the elements
of $\mathsf{Eff}_{+}\left(\mathrm{A}\right)$ are linear functionals
that yield a non-negative number when applied to elements of $\mathsf{St}_{+}\left(\mathrm{A}\right)$.
Therefore, one has $\mathsf{Eff}_{+}\left(\mathrm{A}\right)\subseteq\mathsf{St}_{+}^{*}\left(\mathrm{A}\right)$.
It is interesting to study when one has the equality in this inclusion.
\begin{condition}[No-restriction hypothesis \citep{Chiribella-purification}]
We say that a theory is \emph{non-restricted}, or that it satisfies
the no-restriction hypothesis, if $\mathsf{Eff}_{+}\left(\mathrm{A}\right)=\mathsf{St}_{+}^{*}\left(\mathrm{A}\right)$
for every system.
\end{condition}
While this may look like just a statement of mathematical interest,
it has some important physical implications. Consider the subset of
$\mathsf{St}_{+}^{*}\left(\mathrm{A}\right)$ made of linear functionals
$f$ such that $\left(f\middle|\rho\right)\in\left[0,1\right]$ for
all states $\rho$. In a non-restricted theory, these elements $f$
are also valid effects. In other words, the no-restriction hypothesis
states that every mathematically allowed effect is also a physical
effect. Clearly, the no-restriction hypothesis concerns more the mathematical
structure of the theory than its operational one. Indeed, it is the
duty of the physical theory to specify what objects are to be considered
physical effects, even if they are admissible in principle, based
on their mathematical properties. For this reason, the no-restriction
hypothesis has been questioned various times on the basis of its lack
of operational motivation \citep{Chiribella-purification,No-restriction,No-restriction2}.
Moreover, recently it has been show that theories with almost quantum
correlations \citep{Almost-quantum} violate the no-restriction hypothesis
\citep{Almost-no-restriction}.

Examples of theories that satisfy the no-restriction hypothesis are
classical and quantum theory. The theory of restricted trits in appendix~\ref{sec:restricted trit},
instead, explicitly violates it. This theory also has no classical
states (see appendix~\ref{sec:restricted trit}); this is not a coincidence,
for one of the most important consequences of the no-restriction hypothesis
is that a non-restricted theory always admits at least the classical
bit as a sub-theory.
\begin{prop}
In a non-restricted theory, for every pure state $\psi_{1}$ there
exists another pure state $\psi_{2}$ such that $\left\{ \psi_{1},\psi_{2}\right\} $
are perfectly distinguishable.
\end{prop}
\begin{proof}
Let $\psi_{1}$ be a pure state. The proof will consist of some steps.
In the first step, let us prove that there exists a non-trivial element
$f$ of the dual cone $\mathsf{St}_{+}^{*}\left(\mathrm{A}\right)$
such that $\left(f\middle|\psi_{1}\right)=0$. Note that being pure,
$\psi_{1}$ lies in some supporting hyperplane through the origin
of the cone $\mathsf{St}_{+}\left(\mathrm{A}\right)$ \citep{Boyd}.
Such a hyperplane must have equation $\left(f\middle|x\right)=0$
for all $x\in\mathsf{St}_{\mathbb{R}}\left(\mathrm{A}\right)$, where
$f$ is some non-trivial linear functional on $\mathsf{St}_{\mathbb{R}}\left(\mathrm{A}\right)$,
otherwise it would not pass through the origin (i.e.\ the zero vector).
Being a supporting hyperplane, we can choose $f$ to be in the dual
cone $\mathsf{St}_{+}^{*}\left(\mathrm{A}\right)$ \citep{Boyd}.
Thus we have found $f\in\mathsf{St}_{+}^{*}\left(\mathrm{A}\right)$
such that $\left(f\middle|\psi_{1}\right)=0$.

Let us consider the maximum of $f$ on the state space. Since $f$
is continuous and the state space is compact, it achieves its maximum
$\lambda^{*}$ on some state $\rho^{*}$. Note that $\lambda^{*}>0$,
otherwise $f$ would be the zero functional. Let us show that the
maximum is attained on some pure state. If $\rho^{*}$ is already
a pure state, there is nothing to prove. If it is not, consider a
refinement of $\rho^{*}$ in terms of pure states, $\rho^{*}=\sum_{i}p_{i}\psi_{i}$,
where $\left\{ p_{i}\right\} $ is a probability distribution. Apply
$f$ to $\rho^{*}$: 
\[
\lambda^{*}=\left(f\middle|\rho^{*}\right)=\sum_{i}p_{i}\left(f\middle|\psi_{i}\right).
\]
Clearly $\lambda^{*}\leq\max\left(f\middle|\psi_{i}\right)$, but
being $\lambda^{*}$ the maximum of $f$ , in fact $\lambda^{*}=\max\left(f\middle|\psi_{i}\right)$.
This means that there is a pure state $\psi_{2}$, chosen among these
$\psi_{i}$'s, on which $f$ attains its maximum.

Now consider the functional $a_{2}:=\frac{f}{\lambda^{*}}$, which
takes values in the interval $\left[0,1\right]$ when applied to states.
Specifically $\left(a_{2}\middle|\psi_{2}\right)=1$ and $\left(a_{2}\middle|\psi_{1}\right)=0$.
By the no-restriction hypothesis, it is a valid effect, so we can
construct the measurement $\left\{ a_{1},a_{2}\right\} $, where $a_{1}:=u-a_{2}$,
which perfectly distinguishes between $\psi_{1}$ and $\psi_{2}$.
\end{proof}
The essence of this proposition is that in every non-restricted physical
theory there are at least two perfectly distinguishable pure states.
By possibly adding other pure states so that the overall set is perfectly
distinguishable, one can find a maximal set of perfectly distinguishable
pure states. In this way one can always construct a classical set
for every system, of dimension at least 2.

Even though the no-restriction hypothesis guarantees the existence
of classical sets, we do not wish to assume it for its lack of operational
motivation, preferring to stick to condition~\ref{cond:ECC}, which
is agnostic about the reason why classical states arise in a theory.

\section{A theory with no classical states\label{sec:restricted trit}}

In this appendix we present a theory from which classical theory cannot
emerge in any way, neither through decoherence, nor by taking arbitrarily
large systems. This is the theory of \emph{restricted trits}, which
describes a classical trit (and its composites) when we fundamentally
restrict its possible measurements. Remarkably, the effect of this
restriction is that the theory contains \emph{no classical states,
nor effective approximations of them}.

To construct this theory, we start from the state space of the classical
trit, represented in fig.~\ref{fig:The-state-space}, with pure states
$\alpha_{1}$, $\alpha_{2}$, $\alpha_{3}$. Let $\left\{ a_{1},a_{2},a_{3}\right\} $
be the measurement that perfectly distinguishes them in a single shot:
$\left(a_{i}\middle|\alpha_{j}\right)=\delta_{ij}$. Instead of allowing
the full set of effects of classical theory, suppose that, for some
reason, the most fine-grained effects that are allowed are $e_{ij}=\frac{1}{2}\left(a_{i}+a_{j}\right)$,
with $i<j$. A section of the dual cone (the same as the effect cone
of classical theory), and of the effect cone of the restricted trit
is represented in fig.~\ref{fig:A-section-of} in the main text.

Since we have a smaller set of effects than the original classical
trit, we must check what happens to the state space. Indeed it may
happen that two states become tomographically indistinguishable because
there are not enough effects to witness their difference (cf.\ definition~\ref{def:tomographically-distinct}).
However, this is not the case for the restricted trit, because the
effects $e_{ij}$ are linearly independent. As such, they span exactly
the same effect vector space as the effects $a_{i}$, which is what
determines the tomographic power of a theory. Therefore the state
space of the restricted trit coincides with that of the classical
trit (cf.\ fig.~\ref{fig:The-state-space} in the main text).

Nevertheless, the restriction on the allowed effects has a dramatic
consequence: \emph{there are no perfectly distinguishable pure states,
therefore no classical states even when composing an arbitrary large
number of restricted trits}. To this end, first let us show that $\left\{ \alpha_{1},\alpha_{2},\alpha_{3}\right\} $
are no longer perfectly distinguishable. Consider a generic effect
$e=\lambda_{12}e_{12}+\lambda_{13}e_{13}+\lambda_{23}e_{23}$, where
$\lambda_{ij}\geq0$. This effect could yield 0 on $\alpha_{2}$ and
$\alpha_{3}$ if and only if $\lambda_{12}=\lambda_{13}=\lambda_{23}=0$,
but this would be the zero effect, which cannot yield 1 on $\alpha_{1}$.
This means that the $\alpha_{i}$'s cannot be jointly perfectly distinguishable.

Maybe we can still find a pair of $\alpha_{i}$'s that are perfectly
distinguishable? The answer is again negative. To see it, take e.g.\ the
pair $\left\{ \alpha_{1},\alpha_{2}\right\} $ (for the others the
argument is the same). The only element in the effect cone that yields
1 on $\alpha_{1}$ and 0 on $\alpha_{2}$ is $2e_{13}$, but this
is \emph{not} a physical effect, because $u-2e_{13}=a_{2}$, which
is \emph{not} an effect. In other words, $2e_{13}$ cannot exist in
a measurement of the form $\left\{ 2e_{13},u-2e_{13}\right\} $, but
all effects must be part of some measurement! In conclusion, the restricted
trit has \emph{no} classical states.

What about the other systems of this theory? They are generated by
composing restricted trits using the\emph{ minimal tensor product}
\citep{Namioka1969tensor,Barnum2016}: the normalized states of the
composite system $\mathrm{AB}$ are given by the convex hull of product
states of $\mathrm{A}$ and $\mathrm{B}$:
\[
\mathsf{St}_{1}\left(\mathrm{AB}\right)=\mathrm{Conv}\left\{ \rho_{\mathrm{A}}\otimes\rho_{\mathrm{B}}:\rho_{\mathrm{A}}\in\mathsf{St}_{1}\left(\mathrm{A}\right),\rho_{\mathrm{B}}\in\mathsf{St}_{1}\left(\mathrm{B}\right)\right\} ,
\]
where the subscript 1 means that we are only considering normalized
states. This is how ordinary classical systems compose. Similarly
the effect cones (generated by the extreme effects $e_{ij}$) are
composed using the minimal tensor product of cones, whereby
\begin{equation}
\mathsf{Eff}_{+}\left(\mathrm{AB}\right)=\mathrm{Con}\left\{ e_{ij}\otimes e'_{kl}\right\} ,\label{eq:minimal product effects}
\end{equation}
where $\mathrm{Con}$ denotes the conical hull. The generic composite
system is obtained by composing $N$ restricted trits. Therefore,
it has $3^{N}$ pure states and $3^{N}$ extreme effects, given by
$e_{i_{1}j_{1}}\otimes\ldots\otimes e_{i_{N}j_{N}}$, where for each
$k\in\left\{ 1,\ldots,N\right\} $ one has $i_{k}<j_{k}$, with $i_{k},j_{k}\in\left\{ 1,2,3\right\} $,
and $e_{i_{k}j_{k}}=\frac{1}{2}\left(a_{i_{k}}+a_{j_{k}}\right)$.
Given that the $\left\{ e_{i_{k}j_{k}}\right\} $ are linearly independent
for every $k\in\left\{ 1,\ldots,N\right\} $, the effects $\left\{ e_{i_{1}j_{1}}\otimes\ldots\otimes e_{i_{N}j_{N}}\right\} $
are still linearly independent, therefore they span the same vector
space as the effects $\left\{ a_{i_{1}}\otimes\ldots\otimes a_{i_{N}}\right\} $.
This means that all states will stay tomographically distinct, and
the state space of a composite of $N$ restricted trits will look
like the composite of $N$ classical trits, namely, like a simplex
with $3^{N}$ vertices.

Let us show that, even in composites, we still have a restriction
on the mathematically allowed effects, represented by the dual cone.
To this end, let us show that, for instance, we cannot obtain the
effect $a_{1}^{\otimes N}$ out of conical combinations of the extreme
effects $e_{i_{1}j_{1}}\otimes\ldots\otimes e_{i_{N}j_{N}}$ (for
other products of the $a_{i}$'s the argument is the same). Our goal
is to determine the \emph{non-negative} coefficients $\lambda_{i_{1}j_{1},\ldots,i_{N}j_{N}}$
such that
\[
a_{1}\otimes\ldots\otimes a_{1}=\sum_{i_{1},j_{1}}\ldots\sum_{i_{N},j_{N}}\lambda_{i_{1}j_{1},\ldots,i_{N}j_{N}}e_{i_{1}j_{1}}\otimes\ldots\otimes e_{i_{N}j_{N}}.
\]
Recalling the definition of $e_{i_{k}j_{k}}$ we have
\begin{equation}
a_{1}\otimes\ldots\otimes a_{1}=\frac{1}{2^{N}}\sum_{i_{1},j_{1}}\ldots\sum_{i_{N},j_{N}}\lambda_{i_{1}j_{1},\ldots,i_{N}j_{N}}\left(a_{i_{1}}+a_{j_{1}}\right)\otimes\ldots\otimes\left(a_{i_{N}}+a_{j_{N}}\right).\label{eq:big sum}
\end{equation}
Unfolding the above expression, we get
\[
a_{1}^{\otimes N}=\frac{1}{2^{N}}\sum_{j_{1},\ldots,j_{N}=2}^{3}\lambda_{1j_{1},\ldots,1j_{N}}a_{1}^{\otimes N}+\sum_{i_{1},j_{1}}'\ldots\sum_{i_{N},j_{N}}'\lambda_{i_{1}j_{1},\ldots,i_{N}j_{N}}\left(a_{i_{1}}+a_{j_{1}}\right)\otimes\ldots\otimes\left(a_{i_{N}}+a_{j_{N}}\right),
\]
where primed summations indicate the other summation terms in eq.~\eqref{eq:big sum}.
Clearly, primed summations must vanish, but since $\lambda_{i_{1}j_{1},\ldots,i_{N}j_{N}}\geq0$,
all coefficients $\lambda_{i_{1}j_{1},\ldots,i_{N}j_{N}}$ in primed
summations must vanish. Note that the coefficients $\lambda_{1j_{1},\ldots,1j_{N}}$
arise among the coefficients in the primed summations, which are all
zero. This means that $a_{1}^{\otimes N}=0$, which contradicts the
hypothesis. It follows that it is \emph{not} possible to obtain products
of the $a_{i}$'s from conical combination of the extreme effects
$e_{i_{1}j_{1}}\otimes\ldots\otimes e_{i_{N}j_{N}}$. In other words,
we are still in the presence of a restriction on the set of mathematically
allowed effects.

Let us show that even in composites of the restricted trit there are
no perfectly distinguishable pure states. To this end, it is enough
to show that there are no \emph{pairs} of perfectly distinguishable
pure states. Indeed, if the states $\left\{ \rho_{i}\right\} _{i=1}^{d}$,
with $d>2$, are distinguished by the measurement $\left\{ a_{i}\right\} _{i=1}^{d}$,
any pair $\left\{ \rho_{1},\rho_{2}\right\} \subset\left\{ \rho_{i}\right\} _{i=1}^{d}$
is also a set of perfectly distinguishable states: they are perfectly
distinguished by the measurement $\left\{ a_{1},u-a_{1}\right\} $.
Therefore no pairs of perfectly distinguishable pure states implies
no sets of perfectly distinguishable pure states. Note that, in all
composites, pure states are only of the product form; specifically
they are the states $\alpha_{i_{1}}\otimes\ldots\otimes\alpha_{i_{N}}$
for the composition of $N$ restricted trits. Now we will show by
induction on $N$ that there are no pairs of perfectly distinguishable
pure states in any composite in the theory of restricted trits. For
$N=1$, we have already proved it. Now suppose this is true for $N$,
and let us show that it is valid also for $N+1$. Take system $\mathrm{A}$
to be the composition of $N$ restricted trits, and let $\mathrm{B}$
be a single restricted trit. Suppose by contradiction that system
$\mathrm{AB}$, given by the composition of $N+1$ restricted trits
has two perfectly distinguishable pure states. They must be of the
form $\left\{ \alpha_{1}\otimes\beta_{1},\alpha_{2}\otimes\beta_{2}\right\} $,
where $\left\{ \alpha_{1},\alpha_{2}\right\} $ are pure states of
$\mathrm{A}$, and $\left\{ \beta_{1},\beta_{2}\right\} $ are pure
states of $\mathrm{B}$. Since $\left\{ \alpha_{1}\otimes\beta_{1},\alpha_{2}\otimes\beta_{2}\right\} $
are perfectly distinguishable, there exists a measurement $\left\{ E_{1},E_{2}\right\} $
on $\mathrm{AB}$ such that 
\[
\left\{ \begin{array}{l}
\left(E_{1}\middle|\alpha_{1}\otimes\beta_{1}\right)=1\\
\left(E_{1}\middle|\alpha_{2}\otimes\beta_{2}\right)=0
\end{array}\right.
\]
and
\[
\left\{ \begin{array}{l}
\left(E_{2}\middle|\alpha_{1}\otimes\beta_{1}\right)=0\\
\left(E_{2}\middle|\alpha_{2}\otimes\beta_{2}\right)=1
\end{array}\right..
\]
Now, by eq.~\eqref{eq:minimal product effects}, $E_{1}$ is a conical
combination of products of the extreme effects: $E_{1}=\sum_{i}\lambda_{i,1}a_{i,1}\otimes b_{i,1}$,
with $\lambda_{i,1}\geq0$, where $a_{i,1}$ and $b_{i,1}$ are extreme
effects of $\mathrm{A}$ and $\mathrm{B}$ respectively. Then we have
\[
\left(E_{1}\middle|\alpha_{2}\otimes\beta_{2}\right)=\sum_{i}\lambda_{i,1}\left(a_{i,1}\middle|\alpha_{2}\right)\left(b_{i,1}\middle|\beta_{2}\right)=0.
\]
Note that not all $\lambda_{i,1}$ can be zero, otherwise $E_{1}$
would be the zero vector. Therefore we have two possibilities, which
can be both true at the same time:
\begin{enumerate}
\item $\left(a_{i,1}\middle|\alpha_{2}\right)=0$ for every $i$. In this
case, consider the effect $e_{1}$ of $\mathrm{A}$ defined as $e_{1}:=\sum_{i}\lambda_{i,1}\left(b_{i,1}\middle|\beta_{1}\right)a_{i,1}$,
or in diagrams\[
\begin{aligned}\Qcircuit @C=1em @R=.7em @!R {   & \qw \poloFantasmaCn{\rA} &\measureD{e_1}}\end{aligned}~:=\!\!\!\!\begin{aligned}\Qcircuit @C=1em @R=.7em @!R { &  & \qw \poloFantasmaCn{\rA} &\multimeasureD{1}{E_1} \\ & \prepareC{\beta_1} & \qw \poloFantasmaCn{\rB}  &\ghost{E_1} }\end{aligned}~.
\]We have
\[
\left(e_{1}\middle|\alpha_{1}\right)=\sum_{i}\lambda_{i,1}\left(a_{i,1}\middle|\alpha_{1}\right)\left(b_{i,1}\middle|\beta_{1}\right)=\left(E_{1}\middle|\alpha_{1}\otimes\beta_{1}\right)=1
\]
and
\[
\left(e_{1}\middle|\alpha_{2}\right)=\sum_{i}\lambda_{i,1}\left(a_{i,1}\middle|\alpha_{2}\right)\left(b_{i,1}\middle|\beta_{1}\right)=0
\]
because $\left(a_{i,1}\middle|\alpha_{2}\right)=0$ for every $i$.
This means that the pure states $\left\{ \alpha_{1},\alpha_{2}\right\} $
are perfectly distinguished by the measurement $\left\{ e_{1},u-e_{1}\right\} $.
This contradicts the induction hypothesis that for the composite of
$N$ restricted trits (system $\mathrm{A}$) there are no perfectly
distinguishable pure states.
\item $\left(b_{i,1}\middle|\beta_{2}\right)=0$ for every $i$. The proof
is essentially the same as in the previous case. Consider the effect
$f_{1}$ of $\mathrm{B}$, defined as $f_{1}:=\sum_{i}\lambda_{i,1}\left(a_{i,1}\middle|\alpha_{1}\right)b_{i,1}$,
or in diagrams\[
\begin{aligned}\Qcircuit @C=1em @R=.7em @!R {   & \qw \poloFantasmaCn{\rB} &\measureD{f_1}}\end{aligned}~:=\!\!\!\!\begin{aligned}\Qcircuit @C=1em @R=.7em @!R { &\prepareC{\alpha_1}  & \qw \poloFantasmaCn{\rA} &\multimeasureD{1}{E_1} \\ &  & \qw \poloFantasmaCn{\rB}  &\ghost{E_1} }\end{aligned}~.
\]One has
\[
\left(f_{1}\middle|\beta_{1}\right)=\sum_{i}\lambda_{i,1}\left(a_{i,1}\middle|\alpha_{1}\right)\left(b_{i,1}\middle|\beta_{1}\right)=\left(E_{1}\middle|\alpha_{1}\otimes\beta_{1}\right)=1
\]
 and
\[
\left(f_{1}\middle|\beta_{2}\right)=\sum_{i}\lambda_{i,1}\left(a_{i,1}\middle|\alpha_{1}\right)\left(b_{i,1}\middle|\beta_{2}\right)=0
\]
because $\left(b_{i,1}\middle|\beta_{2}\right)=0$ for every $i$.
Hence the pure states $\left\{ \beta_{1},\beta_{2}\right\} $ are
perfectly distinguished by the measurement $\left\{ f_{1},u-f_{1}\right\} $.
This contradicts the fact that there are no perfectly distinguishable
pure states in the restricted trit (system $\mathrm{B}$).
\end{enumerate}
In conclusion, we have proved that, in all composite systems in the
theory of restricted trits, there are no perfectly distinguishable
pure states. This means that no suitable classical limit can exist
for this theory, not even considering an extremely large system.

\section{Complete decoherence\label{sec:Complete-decoherence}}

From the definition of complete decoherence on a classical set $\boldsymbol{\alpha}$
(definition~\ref{def:complete decoherence} in the main text), it
is immediate to see that applying the same decoherence twice on a
single system is like applying it once. In other words, $D_{\boldsymbol{\alpha}}^{2}\rho=D_{\boldsymbol{\alpha}}\rho$
for every state $\rho$. Indeed, by definition $D_{\boldsymbol{\alpha}}\rho$
is a classical state $\gamma$, and applying the complete decoherence
again, this classical state stays the same.

From a physical point of view, the fact that $D_{\boldsymbol{\alpha}}^{2}\rho=D_{\boldsymbol{\alpha}}\rho$
means that once a (single) system is decohered, classicality is reached,
and there is nothing left to decohere. Note that $D_{\boldsymbol{\alpha}}^{2}\rho=D_{\boldsymbol{\alpha}}\rho$
for every $\rho$ is \emph{not} enough to conclude that $D_{\boldsymbol{\alpha}}^{2}=D_{\boldsymbol{\alpha}}$,
unless the theory satisfies Local Tomography \citep{Chiribella-purification},
because, in general, transformations are defined by their action on
half of a bipartite state, not on a state of a single system (see
appendix~\ref{sec:General-probabilistic-theories}).

After understanding the behavior of complete decoherence on states,
we need to look at what happens if we apply it to the effects of the
theory. As it maps every state to a classical state, we expect that
it does the same with effects: every effect becomes classical. This
is indeed the case, as shown by the following:
\begin{prop}
The set $\left\{ eD_{\boldsymbol{\alpha}}\right\} $ of decohered
effects, where $e\in\mathsf{Eff}\left(\mathrm{A}\right)$ is an effect
of the original theory, coincides with the set of classical effects
of $\boldsymbol{\alpha}$.
\end{prop}
\begin{proof}
To show that the two sets coincide, we will actually show that there
is a canonical bijection between the set of decohered effects $\left\{ eD_{\boldsymbol{\alpha}}\right\} $
and the set $\mathsf{Eff}\left(\mathrm{A}\right)/\boldsymbol{\alpha}$.
This bijection associates the equivalence class $\left[eD_{\boldsymbol{\alpha}}\right]$
in $\mathsf{Eff}\left(\mathrm{A}\right)/\boldsymbol{\alpha}$ with
every decohered effect $eD_{\boldsymbol{\alpha}}$. Let us prove that
this is indeed a bijection. To this end, first observe that two decohered
effects $eD_{\boldsymbol{\alpha}}$ and $fD_{\boldsymbol{\alpha}}$
are equal if and only if $e\sim_{\boldsymbol{\alpha}}f$. Indeed,
$eD_{\boldsymbol{\alpha}}=fD_{\boldsymbol{\alpha}}$ if and only if
$\left(e\middle|D_{\boldsymbol{\alpha}}\middle|\rho\right)=\left(f\middle|D_{\boldsymbol{\alpha}}\middle|\rho\right)$,
for every state $\rho$. Now, define $\gamma:=D_{\boldsymbol{\alpha}}\rho$,
which is a classical state. Therefore, $eD_{\boldsymbol{\alpha}}=fD_{\boldsymbol{\alpha}}$
if and only if $\left(e\middle|\gamma\right)=\left(f\middle|\gamma\right)$
for every classical state $\gamma\in\boldsymbol{\alpha}$, which means
$e\sim_{\boldsymbol{\alpha}}f$.

Let us prove that the mapping $eD_{\boldsymbol{\alpha}}\mapsto\left[eD_{\boldsymbol{\alpha}}\right]$
is injective. Assume $\left[eD_{\boldsymbol{\alpha}}\right]=\left[fD_{\boldsymbol{\alpha}}\right]$,
and let us show that $eD_{\boldsymbol{\alpha}}=fD_{\boldsymbol{\alpha}}$.
If $\left[eD_{\boldsymbol{\alpha}}\right]=\left[fD_{\boldsymbol{\alpha}}\right]$,
then $eD_{\boldsymbol{\alpha}}\sim_{\boldsymbol{\alpha}}fD_{\boldsymbol{\alpha}}$,
which means $\left(e\middle|D_{\boldsymbol{\alpha}}\middle|\gamma\right)=\left(f\middle|D_{\boldsymbol{\alpha}}\middle|\gamma\right)$
for every classical state $\gamma\in\boldsymbol{\alpha}$. Now, $D_{\boldsymbol{\alpha}}\gamma=\gamma$,
so we have $\left(e\middle|\gamma\right)=\left(f\middle|\gamma\right)$
for every $\gamma\in\boldsymbol{\alpha}$. This means $e\sim_{\boldsymbol{\alpha}}f$,
which allows us to conclude that $eD_{\boldsymbol{\alpha}}=fD_{\boldsymbol{\alpha}}$.

Now, let us prove that the mapping $eD_{\boldsymbol{\alpha}}\mapsto\left[eD_{\boldsymbol{\alpha}}\right]$
is surjective too. Take any equivalence class $\left[e\right]$ in
$\mathsf{Eff}\left(\mathrm{A}\right)/\boldsymbol{\alpha}$, and let
us show that $\left[e\right]=\left[eD_{\boldsymbol{\alpha}}\right]$,
so the equivalence class $\left[e\right]$ is associated with the
decohered effect $eD_{\boldsymbol{\alpha}}$. Now, for every classical
state $\gamma\in\boldsymbol{\alpha}$, $\left(e\middle|D_{\boldsymbol{\alpha}}\middle|\gamma\right)=\left(e\middle|\gamma\right)$
because $D_{\boldsymbol{\alpha}}\gamma=\gamma$. This shows that $e\sim_{\boldsymbol{\alpha}}eD_{\boldsymbol{\alpha}}$,
so $\left[e\right]=\left[eD_{\boldsymbol{\alpha}}\right]$.

Finally, let us show that the mapping $eD_{\boldsymbol{\alpha}}\mapsto\left[eD_{\boldsymbol{\alpha}}\right]$
respects the sums defined in the respective sets. Suppose $e+f$ is
a valid effect of the theory, where $e$ and $f$ are two valid effects,
then $eD_{\boldsymbol{\alpha}}+fD_{\boldsymbol{\alpha}}=\left(e+f\right)D_{\boldsymbol{\alpha}}$
is a valid decohered effect. With this effect we associate the equivalence
class $\left[\left(e+f\right)D_{\boldsymbol{\alpha}}\right]=\left[eD_{\boldsymbol{\alpha}}\right]+\left[fD_{\boldsymbol{\alpha}}\right]$,
which shows that the mapping $eD_{\boldsymbol{\alpha}}\mapsto\left[eD_{\boldsymbol{\alpha}}\right]$
respects the sums. This allows us to conclude that all classical effects
can be regarded as decohered effects and vice versa.
\end{proof}
Thus, complete decoherence maps all effects to classical effects,
but it is not obvious if it leaves classical effects invariant: in
general, it just sends them to an equivalent effect on $\boldsymbol{\alpha}$.
From this point of view, the most general definition of complete decoherence
(definition~\ref{def:complete decoherence} in the main text) is
asymmetric since classical states are left invariant, but not classical
effects in general.

Definition~\ref{def:complete decoherence} in the main text is so
general that, in principle, given classical set $\boldsymbol{\alpha}$,
there may be \emph{more than one} channel that is a complete decoherence
on $\boldsymbol{\alpha}$. There are essentially two possible ways
in which a complete decoherence on $\boldsymbol{\alpha}$ might be
non-unique.
\begin{enumerate}
\item We can have two complete decoherences on $\boldsymbol{\alpha}$, $D_{1,\boldsymbol{\alpha}}$
and $D_{2,\boldsymbol{\alpha}}$, that decohere some state $\rho$
to different classical states: $D_{1,\boldsymbol{\alpha}}\rho\neq D_{2,\boldsymbol{\alpha}}\rho$.
\item More subtly, if the theory does \emph{not} satisfy Local Tomography,
two complete decoherences $D_{1,\boldsymbol{\alpha}}$ and $D_{2,\boldsymbol{\alpha}}$
on $\boldsymbol{\alpha}$ can be indistinguishable at the level of
single systems, namely $D_{1,\boldsymbol{\alpha}}\rho=D_{2,\boldsymbol{\alpha}}\rho$
for every $\rho$, but they can differ when applied only to part of
a bipartite state:\[
\begin{aligned}\Qcircuit @C=1em @R=.7em @!R { & \multiprepareC{1}{\rho} & \qw \poloFantasmaCn{\rA} & \gate{D_{1,\boldsymbol{\alpha}}} & \qw \poloFantasmaCn{\rA} & \qw \\ & \pureghost{\rho} & \qw \poloFantasmaCn{\rB} & \qw &\qw &\qw}\end{aligned}~\neq\!\!\!\!\begin{aligned}\Qcircuit @C=1em @R=.7em @!R { & \multiprepareC{1}{\rho} & \qw \poloFantasmaCn{\rA} & \gate{D_{2,\boldsymbol{\alpha}}} & \qw \poloFantasmaCn{\rA} & \qw \\ & \pureghost{\rho} & \qw \poloFantasmaCn{\rB} & \qw &\qw &\qw}\end{aligned}~.
\]
\end{enumerate}
For quantum theory, however, definition~\ref{def:complete decoherence}
in the main text is enough to pick a unique decoherence for every
fixed orthonormal basis, which is clearly the TID on that basis. The
proof is not included here since is not relevant for the present paper.
In appendix~\ref{subsec:Test-induced-decoherence-and}, we present
an example of a GPT in which the decoherence on a classical set is
highly non-unique (cf.\ example~\ref{exa:uniqueness}), in clear
contrast with the quantum behavior.

\subsection{Test-induced decoherence and its properties\label{subsec:Test-induced-decoherence-and}}

The next relevant question we investigate is whether a complete decoherence
actually exists in every fundamental causal theory. In quantum theory,
the decoherence on a classical set $\boldsymbol{\alpha}$, described
by an orthonormal basis $\left\{ \left|\alpha_{j}\right\rangle \right\} _{j=1}^{d}$,
is obtained from the von Neumann measurement on that orthonormal basis.
Indeed, if we sum over all outcomes, we get the complete decoherence.

In appendix~\ref{subsec:Consequences-of-Causality}, we noted that
in causal theories one can always construct measure-and-prepare tests.
Now we build one out of the pure states $\left\{ \alpha_{i}\right\} _{i=1}^{d}$
of a classical set and their associated distinguishing measurement
$\left\{ a_{i}\right\} _{i=1}^{d}$: the test $\left\{ \left|\alpha_{i}\right)\left(a_{i}\right|\right\} _{i=1}^{d}$,
which can be viewed as a non-demolition measurement on the classical
set $\boldsymbol{\alpha}$. Taking the coarse-graining over all $d$
outcomes yields a measure-and-prepare channel $\widehat{D}_{\boldsymbol{\alpha}}=\sum_{i=1}^{d}\left|\alpha_{i}\right)\left(a_{i}\right|$.
It is straightforward to show that $\widehat{D}_{\boldsymbol{\alpha}}$
is a complete decoherence, which we term the \emph{test-induced decoherence}
(TID).
\begin{proof}[Proof of proposition~\ref{prop:TID} in the main text]
We must check if $\widehat{D}_{\boldsymbol{\alpha}}$ satisfies the
two properties defining a complete decoherence (cf.\ definition~\ref{def:complete decoherence}
in the main text).
\begin{enumerate}
\item For any state $\rho$,
\[
\widehat{D}_{\boldsymbol{\alpha}}\rho=\sum_{i=1}^{d}\left|\alpha_{i}\right)\left(a_{i}\middle|\rho\right)=:\sum_{i=1}^{d}p_{i}\alpha_{i},
\]
where we have set $p_{i}:=\left(a_{i}\middle|\rho\right)$. Note that
$p_{i}\in\left[0,1\right]$, and that $\sum_{i=1}^{d}p_{i}=1$ because
$\left\{ a_{i}\right\} _{i=1}^{d}$ is a measurement. Therefore $\widehat{D}_{\boldsymbol{\alpha}}\rho$
is a classical state, lying in the simplex generated by the $\alpha_{i}$'s.
\item For any $\alpha_{i}$, we have
\[
\widehat{D}_{\boldsymbol{\alpha}}\alpha_{i}=\sum_{j=1}^{d}\left|\alpha_{j}\right)\left(a_{j}\middle|\alpha_{i}\right)=\alpha_{i}.
\]
This means that $\widehat{D}_{\boldsymbol{\alpha}}$ preserves all
pure states in $\boldsymbol{\alpha}$, and by linearity it preserves
all classical states in $\boldsymbol{\alpha}$.
\end{enumerate}
\end{proof}
This shows that a complete decoherence always exists in causal theories.
The TID enjoys some remarkable properties that make it a physically
motivated form of decoherence. 
\begin{prop}
Let $\left\{ a_{i}\right\} _{i=1}^{d}$ be the measurement associated
with the classical set $\boldsymbol{\alpha}$. The TID $\widehat{D}_{\boldsymbol{\alpha}}$
satisfies the following properties:
\begin{enumerate}
\item $\widehat{D}_{\boldsymbol{\alpha}}^{2}=\widehat{D}_{\boldsymbol{\alpha}}$.\label{enu:idempotence}
\item $a_{i}\widehat{D}_{\boldsymbol{\alpha}}=a_{i}$ for every $i$.\label{enu:classical effects respected}
\end{enumerate}
\end{prop}
\begin{proof}
Let us prove the two properties.
\begin{enumerate}
\item Let us compose the TID with itself:
\[
\widehat{D}_{\boldsymbol{\alpha}}^{2}=\widehat{D}_{\boldsymbol{\alpha}}\sum_{i=1}^{d}\left|\alpha_{i}\right)\left(a_{i}\right|=\sum_{i=1}^{d}\sum_{j=1}^{d}\left|\alpha_{j}\right)\left(a_{j}\middle|\alpha_{i}\right)\left(a_{i}\right|=\sum_{i=1}^{d}\left|\alpha_{i}\right)\left(a_{i}\right|=\widehat{D}_{\boldsymbol{\alpha}}.
\]
\item It is a straightforward calculation. Indeed 
\[
a_{i}\widehat{D}_{\boldsymbol{\alpha}}=\sum_{j=1}^{d}\left(a_{i}\middle|\alpha_{j}\right)\left(a_{j}\right|=a_{i}.
\]
\end{enumerate}
\end{proof}
Note that property~\ref{enu:idempotence} means that the TID satisfies
a stronger idempotence property than a generic complete decoherence:
not only is this property valid on single systems, but also when the
TID is applied to a part of a bipartite state. Again, this means that
to decohere a system completely, it is enough to apply the TID just
once: further applications of the TID will not change anything. Property~\ref{enu:classical effects respected}
states that the TID preserves the effects that perfectly distinguish
the pure states in $\boldsymbol{\alpha}$. Since all classical effects
arise as suitable conical combinations of these effects, it means
that the TID preserves each classical effect. This property removes
the asymmetry we observed in the behavior of complete decoherences,
which in general preserve only classical states, but not classical
effects. The TID, instead, treats classical states and effects on
equal footing, doing nothing to both of them. This makes it more physically
appealing.

Recall that, in quantum theory, decoherence is always associated with
the presence of an environment where information is leaked \citep{Decoherence-review,Review-decoherence,Selby-leaks}.
Instead in definition~\ref{def:complete decoherence} in the main
text, as well as in some other proposals in the GPT literature \citep{Selby-entanglement2,Hyperdecoherence},
the environment does not seem to play any explicit role in the process.
However, in the TID, the environment and external observers are again
present, albeit implicitly. Indeed, the fact that the TID arises as
the coarse-graining of a test means that, at least in principle, an
external observer is present in the process of decoherence.

Previous contributions on decoherence in GPTs \citep{Selby-entanglement2,Hyperdecoherence}
required the complete decoherence to be strictly purity-decreasing
\citep{Selby-entanglement2}, or alternatively, that if a decohered
state is pure, the original state was pure too \citep{Hyperdecoherence}.
In the following counterexample we show that the TID does not satisfy
these desiderata in general: in some theories mixed states can be
decohered to pure states. This behavior sharply contrasts with the
one observed in quantum theory.
\begin{example}
\label{exa:purity-decoherence}Let us consider the square bit \citep{Barrett}.
Here the state space is a square, and the pure states are its vertices.
This theory satisfies the no-restriction hypothesis, so all mathematically
allowed effects are valid effects. The pure states are the vertices
of the square. Fig.~\ref{fig:square-classical} shows the state space.
\begin{figure}
\begin{centering}
\includegraphics[scale=0.7]{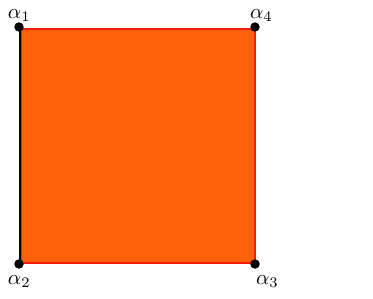}
\par\end{centering}
\caption{\label{fig:square-classical}The state space of the square bit. Here
$\alpha_{1}$, $\alpha_{2}$, $\alpha_{3}$, and $\alpha_{4}$ are
pure states. The classical set $\boldsymbol{\alpha}=\mathrm{Conv}\left\{ \alpha_{1},\alpha_{2}\right\} $
is shown in black.}
\end{figure}
 The pure states are the vectors
\[
\alpha_{1}=\left(\begin{array}{c}
-1\\
1\\
1
\end{array}\right)\qquad\alpha_{2}=\left(\begin{array}{c}
-1\\
-1\\
1
\end{array}\right)\qquad\alpha_{3}=\left(\begin{array}{c}
1\\
-1\\
1
\end{array}\right)\qquad\alpha_{4}=\left(\begin{array}{c}
1\\
1\\
1
\end{array}\right),
\]
where the third component represents the fact that these states are
normalized. Now consider the effects
\[
a_{1}=\frac{1}{2}\left(\begin{array}{ccc}
0 & 1 & 1\end{array}\right)\qquad a_{2}=\frac{1}{2}\left(\begin{array}{ccc}
0 & -1 & 1\end{array}\right).
\]
They make up a measurement $\left\{ a_{1},a_{2}\right\} $ that perfectly
distinguishes the pure states $\left\{ \alpha_{1},\alpha_{2}\right\} $
in a single shot. Therefore, we can consider the classical set $\boldsymbol{\alpha}=\mathrm{Conv}\left\{ \alpha_{1},\alpha_{2}\right\} $,
which is simply the segment connecting $\alpha_{1}$ and $\alpha_{2}$
(see fig.~\ref{fig:square-classical}). Now consider the TID $\widehat{D}_{\boldsymbol{\alpha}}=\left|\alpha_{1}\right)\left(a_{1}\right|+\left|\alpha_{2}\right)\left(a_{2}\right|$.

Note that $\widehat{D}_{\boldsymbol{\alpha}}$ decoheres the mixed
state $\frac{1}{2}\left(\alpha_{1}+\alpha_{4}\right)$ to the pure
state $\alpha_{1}$. This TID definitely increases purity! Is $\frac{1}{2}\left(\alpha_{1}+\alpha_{4}\right)$
the only state with this unexpected behavior? To get a better understanding
let us find out what $\widehat{D}_{\boldsymbol{\alpha}}$ does to
all states of the square bit. To this end, a particularly useful way
to parametrize a generic state of this theory is suggested in fig.~\ref{fig:square}.
\begin{figure}
\begin{centering}
\includegraphics[scale=0.7]{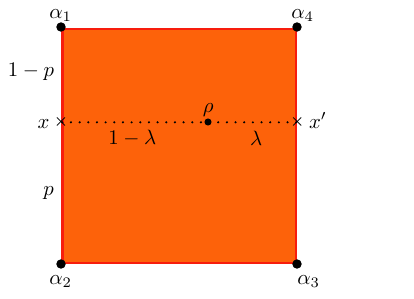}
\par\end{centering}
\caption{\label{fig:square}A particularly useful parametrization of a generic
state $\rho$ of the square bit.}
\end{figure}
 In this parametrization $\rho=\lambda x+\left(1-\lambda\right)x'$,
with $\lambda\in\left[0,1\right]$. Here $x=p\alpha_{1}+\left(1-p\right)\alpha_{2}$,
and $x'=p\alpha_{4}+\left(1-p\right)\alpha_{3}$, where $p\in\left[0,1\right]$.
In summary,
\begin{equation}
\rho=\lambda p\alpha_{1}+\lambda\left(1-p\right)\alpha_{2}+\left(1-\lambda\right)\left(1-p\right)\alpha_{3}+p\left(1-\lambda\right)\alpha_{4}.\label{eq:paramtrization}
\end{equation}
From this expression, it is immediate to see that
\[
\widehat{D}_{\boldsymbol{\alpha}}\rho=p\alpha_{1}+\left(1-p\right)\alpha_{2}=x;
\]
in other words, the TID horizontally projects all states to the $x$-component
of their above parametrization, which belongs to the set $\boldsymbol{\alpha}$.
This is illustrated in fig.~\ref{fig:decoherence}. 
\begin{figure}
\begin{centering}
\includegraphics[viewport=0bp 0bp 159bp 143bp,scale=0.7]{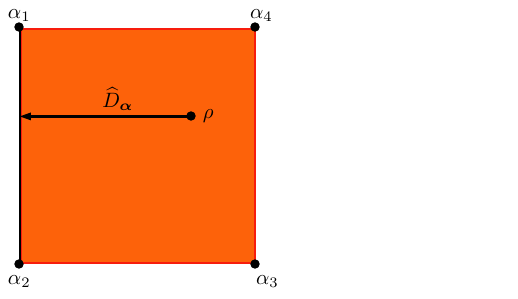}
\par\end{centering}
\caption{\label{fig:decoherence}The action of $\widehat{D}_{\boldsymbol{\alpha}}$
on a state $\rho$ of the square bit is represented as an arrow. The
tip of the arrow indicates the decohered state.}
\end{figure}
From this geometric picture, it is clear that $\widehat{D}_{\boldsymbol{\alpha}}$
decoheres all the mixed states of the form $p\alpha_{1}+\left(1-p\right)\alpha_{4}$,
and $p\alpha_{2}+\left(1-p\right)\alpha_{3}$, with $p\in\left(0,1\right)$,
to pure states ($\alpha_{1}$ and $\alpha_{2}$, respectively).
\end{example}
The natural question is when this counter-intuitive, purity-increasing
behavior of the TID can be observed in a physical theory. In general,
it is enough that one of the distinguishing effects $\left\{ a_{i}\right\} _{i=1}^{d}$,
say $a_{1}$, gives 1 on another pure state $\psi$ \emph{not} in
the classical set $\boldsymbol{\alpha}$. To show it, first note that
since $\left\{ a_{i}\right\} _{i=1}^{d}$ is a measurement, if $\left(a_{1}\middle|\psi\right)=1$,
then $\left(a_{i}\middle|\psi\right)=0$ for $i>1$. Now take any
\emph{mixed} state of the form $p\alpha_{1}+\left(1-p\right)\psi$,
with $p\in\left(0,1\right)$; the TID $\widehat{D}_{\boldsymbol{\alpha}}=\sum_{i=1}^{d}\left|\alpha_{i}\right)\left(a_{i}\right|$
decoheres it to the \emph{pure} state $\alpha_{1}$. In example~\ref{exa:purity-decoherence},
$a_{1}$ gave 1 also on $\alpha_{4}$, which was \emph{not} in the
classical set. Similarly, $a_{2}$ yielded 1 on $\alpha_{3}$ too,
again \emph{not} in the classical set.

Finally, using the toy model of the square bit, we can study the uniqueness
of complete decoherence on some classical sets, showing that, for
some of them, the decoherence is unique (and therefore TID), while
in others it is highly non-unique. This is another important illustration
of how different GPTs can be from the quantum case.
\begin{example}
\label{exa:uniqueness}Consider the classical set $\boldsymbol{\alpha}=\mathrm{Conv}\left\{ \alpha_{1},\alpha_{2}\right\} $
in example~\ref{exa:purity-decoherence} again. Now we prove that
the TID $\widehat{D}_{\boldsymbol{\alpha}}=\left|\alpha_{1}\right)\left(a_{1}\right|+\left|\alpha_{2}\right)\left(a_{2}\right|$
is the only complete decoherence for that classical set. To this end,
let us consider a generic transformation $D$ on the square bit, which
can be represented as a square matrix of order 3:
\[
D=\left(\begin{array}{ccc}
d_{11} & d_{12} & d_{13}\\
d_{21} & d_{22} & d_{23}\\
d_{31} & d_{32} & d_{33}
\end{array}\right).
\]
We want this matrix to represent a complete decoherence $D_{\boldsymbol{\alpha}}$
on $\boldsymbol{\alpha}$. The first condition is to require it to
be a channel; therefore $uD=u$, where $u=\left(\begin{array}{ccc}
0 & 0 & 1\end{array}\right)$ is the deterministic effect (it yields 1 on all the pure states presented
in example~\ref{exa:purity-decoherence}). This condition implies
\begin{equation}
D_{\boldsymbol{\alpha}}=\left(\begin{array}{ccc}
d_{11} & d_{12} & d_{13}\\
d_{21} & d_{22} & d_{23}\\
0 & 0 & 1
\end{array}\right).\label{eq:generic decoherence}
\end{equation}
$D_{\boldsymbol{\alpha}}$ being a complete decoherence on $\boldsymbol{\alpha}$,
we have $D_{\boldsymbol{\alpha}}\alpha_{1}=\alpha_{1}$, $D_{\boldsymbol{\alpha}}\alpha_{2}=\alpha_{2}$,
and $D_{\boldsymbol{\alpha}}\alpha_{3}=p\alpha_{1}+\left(1-p\right)\alpha_{2}$,
for $p\in\left[0,1\right]$, and $D_{\boldsymbol{\alpha}}\alpha_{4}=q\alpha_{1}+\left(1-q\right)\alpha_{2}$,
for $q\in\left[0,1\right]$. Recalling the expression of the pure
states in example~\ref{exa:purity-decoherence}, the conditions $D_{\boldsymbol{\alpha}}\alpha_{1}=\alpha_{1}$,
$D_{\boldsymbol{\alpha}}\alpha_{2}=\alpha_{2}$, and $D_{\boldsymbol{\alpha}}\alpha_{3}=p\alpha_{1}+\left(1-p\right)\alpha_{2}$
yield the linear systems
\[
\left\{ \begin{array}{l}
-d_{11}+d_{12}+d_{13}=-1\\
-d_{11}-d_{12}+d_{13}=-1\\
d_{11}-d_{12}+d_{13}=-1
\end{array}\right.\qquad\left\{ \begin{array}{l}
-d_{21}+d_{22}+d_{23}=1\\
-d_{21}-d_{22}+d_{23}=-1\\
d_{21}-d_{22}+d_{23}=2p-1
\end{array}\right.,
\]
with $p\in\left[0,1\right]$. Solving them, we find that
\[
D_{\boldsymbol{\alpha}}=\left(\begin{array}{ccc}
0 & 0 & -1\\
p & 1 & p\\
0 & 0 & 1
\end{array}\right)
\]
for $p\in\left[0,1\right]$. Let us see if this matrix is compatible
with the condition $D_{\boldsymbol{\alpha}}\alpha_{4}=q\alpha_{1}+\left(1-q\right)\alpha_{2}$
for $q\in\left[0,1\right]$. We have
\[
q\alpha_{1}+\left(1-q\right)\alpha_{2}=q\left(\begin{array}{c}
-1\\
1\\
1
\end{array}\right)+\left(1-q\right)\left(\begin{array}{c}
-1\\
-1\\
1
\end{array}\right)=\left(\begin{array}{c}
-1\\
2q-1\\
1
\end{array}\right).
\]
On the other hand,
\[
D_{\boldsymbol{\alpha}}\alpha_{4}=\left(\begin{array}{ccc}
0 & 0 & -1\\
p & 1 & p\\
0 & 0 & 1
\end{array}\right)\left(\begin{array}{c}
1\\
1\\
1
\end{array}\right)=\left(\begin{array}{c}
-1\\
2p+1\\
1
\end{array}\right).
\]
This means that $2q-1=2p+1$, which means $q=p+1$. The only case
in which $q\in\left[0,1\right]$, when $p\in\left[0,1\right]$, is
when $p=0$. This means that we have a unique complete decoherence
on $\boldsymbol{\alpha}$, which is
\[
D_{\boldsymbol{\alpha}}=\left(\begin{array}{ccc}
0 & 0 & -1\\
0 & 1 & 0\\
0 & 0 & 1
\end{array}\right).
\]
This coincides with the TID $\widehat{D}_{\boldsymbol{\alpha}}=\left|\alpha_{1}\right)\left(a_{1}\right|+\left|\alpha_{2}\right)\left(a_{2}\right|$,
as it is easy to check. This means that on the classical set $\boldsymbol{\alpha}=\mathrm{Conv}\left\{ \alpha_{1},\alpha_{2}\right\} $
there is only one complete decoherence, which is exactly the TID.
By a symmetry argument, we have the same situation whenever we take
a classical set corresponding to a side of the square.

Something completely different, instead, happens when we take the
classical set to be a diagonal of the square. Take e.g.\ $\boldsymbol{\alpha}'=\mathrm{Conv}\left\{ \alpha_{1},\alpha_{3}\right\} $
(fig.~\ref{fig:square diagonal}).
\begin{figure}
\begin{centering}
\includegraphics[scale=0.7]{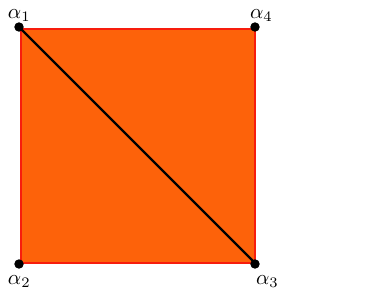}
\par\end{centering}
\caption{\label{fig:square diagonal}Another classical set $\boldsymbol{\alpha}'=\mathrm{Conv}\left\{ \alpha_{1},\alpha_{3}\right\} $
in the square bit, shown in black.}
\end{figure}
 Now we show that in this case we can find an uncountable number of
TIDs! To see it, let us characterize all the distinguishing measurements
for $\left\{ \alpha_{1},\alpha_{3}\right\} $. Consider a generic
effect $e=\left(\begin{array}{ccc}
e_{1} & e_{2} & e_{3}\end{array}\right)$; to be part of a distinguishing measurement, without loss of generality
we can assume $\left(e\middle|\alpha_{1}\right)=1$ and $\left(e\middle|\alpha_{3}\right)=0$.
These conditions imply that $e$ must be of the form
\[
e=\left(\begin{array}{ccc}
e_{2}-\frac{1}{2}, & e_{2}, & \frac{1}{2}\end{array}\right).
\]
This is not enough to guarantee that this is indeed an effect, because
it must give a valid probability on $\alpha_{2}$ and $\alpha_{4}$
as well. In other words, we must impose that $\left(e\middle|\alpha_{2}\right)\in\left[0,1\right]$
and $\left(e\middle|\alpha_{4}\right)\in\left[0,1\right]$. This gives
the following system of inequalities
\[
\left\{ \begin{array}{l}
-2e_{2}+1\geq0\\
-2e_{2}+1\leq1\\
2e_{2}\geq0\\
2e_{2}\leq1
\end{array}\right.,
\]
where the solution is $e_{2}\in\left[0,\frac{1}{2}\right]$. For these
values of $e_{2}$, $e$ is both a mathematically and a physically
allowed effect, because the no-restriction hypothesis is assumed for
the square bit \citep{Barrett}. Similarly, the effect 
\[
e'=u-e=\left(\begin{array}{ccc}
-e_{2}+\frac{1}{2}, & -e_{2}, & \frac{1}{2}\end{array}\right)
\]
is also a mathematically allowed effect. Therefore $\left\{ e,e'\right\} $
is a distinguishing measurement for $\left\{ \alpha_{1},\alpha_{3}\right\} $
for any $e_{2}\in\left[0,\frac{1}{2}\right]$. This gives rise to
a family of TIDs on the classical set $\boldsymbol{\alpha}'=\left\{ \alpha_{1},\alpha_{3}\right\} $
parameterized by a continuous parameter in $\left[0,\frac{1}{2}\right]$:
\begin{equation}
\widehat{D}_{\boldsymbol{\alpha}',t}=\left|\alpha_{1}\right)\left(e_{t}\right|+\left|\alpha_{3}\right)\left(e'_{t}\right|,\label{eq:general MIDs}
\end{equation}
where we have set $t:=e_{2}$ for simplicity of notation, and $t\in\left[0,\frac{1}{2}\right]$.
Using the parametrization of a generic state in eq.~\eqref{eq:paramtrization},
as depicted in fig.~\ref{fig:new parametrization}, we can exemplify
the behavior of the family of TIDs with the two extreme cases of $t=0$
and $t=\frac{1}{2}$.
\begin{figure}
\begin{centering}
\includegraphics[scale=0.7]{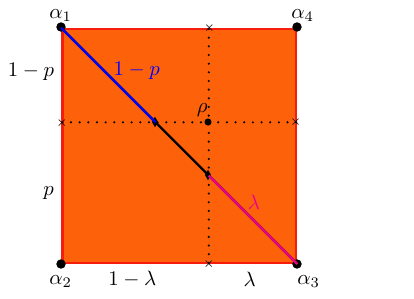}
\par\end{centering}
\caption{\label{fig:new parametrization}Using the intercept theorem, we can
map the coefficients of convex combinations from the sides of the
square to its diagonal. The blue segment is $1-p$ times the diagonal,
while the magenta segment is $\lambda$ times the diagonal.}
\end{figure}
\[
\widehat{D}_{\boldsymbol{\alpha}',0}\rho=\lambda\alpha_{1}+\left(1-\lambda\right)\alpha_{3}\qquad\widehat{D}_{\boldsymbol{\alpha}',\frac{1}{2}}\rho=p\alpha_{1}+\left(1-p\right)\alpha_{3}
\]
This is illustrated in fig.~\ref{fig:double decoherence}.
\begin{figure}
\begin{centering}
\includegraphics[scale=0.7]{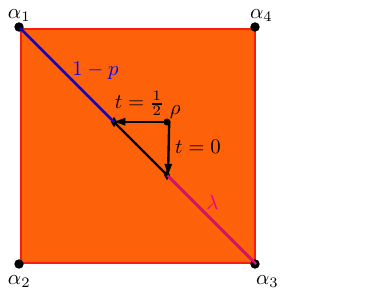}
\par\end{centering}
\caption{\label{fig:double decoherence}The action of the two extreme TIDs
$\widehat{D}_{\boldsymbol{\alpha}',0}$ and $\widehat{D}_{\boldsymbol{\alpha}',\frac{1}{2}}$
on a generic state of the square bit. Again, the tip of the arrow
represents the decohered state.}
\end{figure}
 This means that $\widehat{D}_{\boldsymbol{\alpha}',0}$ projects
every state onto the diagonal along the vertical sides, whereas $\widehat{D}_{\boldsymbol{\alpha}',\frac{1}{2}}$
projects every state onto the diagonal along the horizontal sides.
Even in this case, there are no other complete decoherences besides
the TIDs. Indeed, the generic matrix is like in eq.~\eqref{eq:generic decoherence}.
This time we require $D_{\boldsymbol{\alpha}'}\alpha_{1}=\alpha_{1}$,
$D_{\boldsymbol{\alpha}'}\alpha_{3}=\alpha_{3}$, and $D_{\boldsymbol{\alpha}'}\alpha_{2}=p\alpha_{1}+\left(1-p\right)\alpha_{3}$,
for $p\in\left[0,1\right]$, and $D_{\boldsymbol{\alpha}'}\alpha_{4}=q\alpha_{1}+\left(1-q\right)\alpha_{3}$,
for $q\in\left[0,1\right]$. From the conditions $D_{\boldsymbol{\alpha}'}\alpha_{1}=\alpha_{1}$,
$D_{\boldsymbol{\alpha}'}\alpha_{3}=\alpha_{3}$, and $D_{\boldsymbol{\alpha}'}\alpha_{2}=p\alpha_{1}+\left(1-p\right)\alpha_{3}$,
for $p\in\left[0,1\right]$, we obtain the two linear systems
\[
\left\{ \begin{array}{l}
-d_{11}+d_{12}+d_{13}=-1\\
d_{11}-d_{12}+d_{13}=1\\
-d_{11}-d_{12}+d_{13}=-2p+1
\end{array}\right.\qquad\left\{ \begin{array}{l}
-d_{21}+d_{22}+d_{23}=1\\
d_{21}-d_{22}+d_{23}=-1\\
-d_{21}-d_{22}+d_{23}=2p-1
\end{array}\right.,
\]
for $p\in\left[0,1\right]$. Solving them, we find that
\begin{equation}
D_{\boldsymbol{\alpha}'}=\left(\begin{array}{ccc}
p & p-1 & 0\\
-p & -p+1 & 0\\
0 & 0 & 1
\end{array}\right),\label{eq:general =00005Calpha'}
\end{equation}
for $p\in\left[0,1\right]$. Let us check if this matrix is compatible
with the condition $D_{\boldsymbol{\alpha}}\alpha_{4}=q\alpha_{1}+\left(1-q\right)\alpha_{3}$
for $q\in\left[0,1\right]$. We have
\[
q\alpha_{1}+\left(1-q\right)\alpha_{3}=q\left(\begin{array}{c}
-1\\
1\\
1
\end{array}\right)+\left(1-q\right)\left(\begin{array}{c}
1\\
-1\\
1
\end{array}\right)=\left(\begin{array}{c}
-2q+1\\
2q-1\\
1
\end{array}\right).
\]
On the other hand
\[
D_{\boldsymbol{\alpha}'}\alpha_{4}=\left(\begin{array}{ccc}
p & p-1 & 0\\
-p & -p+1 & 0\\
0 & 0 & 1
\end{array}\right)\left(\begin{array}{c}
1\\
1\\
1
\end{array}\right)=\left(\begin{array}{c}
2p-1\\
-2p+1\\
1
\end{array}\right).
\]
This implies that $2p-1=-2q+1$, whence $q=1-p$. If $p\in\left[0,1\right]$,
this guarantees that $q\in\left[0,1\right]$ too. There are no other
constraints, so the most general complete decoherence on $\boldsymbol{\alpha}'$
is given by the matrix~\eqref{eq:general =00005Calpha'}, with $p\in\left[0,1\right]$.
A straightforward check shows that the TIDs in eq.~\eqref{eq:general MIDs}
with $t\in\left[0,\frac{1}{2}\right]$ cover all the complete decoherences
in eq.~\eqref{eq:general =00005Calpha'} once we set $p:=-2t+1$.
This means that there are no complete decoherences on $\boldsymbol{\alpha}'$
other than the TIDs.
\end{example}

\section{Objectivity in general physical theories}

In this appendix we study the necessary ingredients to prove our main
result: the universality of the form of objective states across all
causal theories with classical states.

\subsection{Properties of sharply repeatable tests\label{sec:Sharply-repeatable-tests}}

A useful property for our research is that SRTs are stable under parallel
composition: if $\left\{ P_{i}\right\} $ and $\left\{ Q_{j}\right\} $
are SRTs, then $\left\{ P_{i}\otimes Q_{j}\right\} $ is also still
an SRT. Another important fact is that SRTs generate perfectly distinguishable
states.
\begin{lem}
\label{lem:distinguishable SRM}Let $\left\{ P_{i}\right\} _{i\in\mathsf{X}}$
be an SRT, and $\rho$ be a generic state, possibly not normalized.
Then if the subset of non-vanishing states of $\left\{ P_{i}\rho\right\} _{i\in\mathsf{X}}$
contains more than one element, these elements can be renormalized
so they are perfectly distinguishable.
\end{lem}
\begin{proof}
Let $\mathsf{I}$ be the subset of $\mathsf{X}$ of the indices labeling
the non-vanishing elements in $\left\{ P_{i}\rho\right\} _{i\in\mathsf{X}}$.
First, let us show that $\mathsf{I}$ is always non-empty. Suppose
by contradiction that it is empty, then $P_{i}\rho=0$ for every $i$.
By Causality, $u=\sum_{i\in\mathsf{X}}uP_{i}$, so 
\[
\left(u\middle|\rho\right)=\sum_{i\in\mathsf{X}}\left(u\middle|P_{i}\middle|\rho\right)=0,
\]
which is impossible on a physical state. Suppose now $\left|\mathsf{I}\right|>1$.
Let us first renormalize the states $\left\{ P_{i}\rho\right\} _{i\in\mathsf{I}}$
by considering $\frac{P_{i}\rho}{\left(u\middle|P_{i}\middle|\rho\right)}$.
Let us prove that they are perfectly distinguished by the measurement
$\left\{ a_{i}\right\} _{i\in\mathsf{I}}$, where
\[
a_{i}=\begin{cases}
uP_{i} & i\neq i_{0}\\
uP_{i_{0}}+\sum_{i\notin\mathsf{I}}uP_{i} & i=i_{0}
\end{cases},
\]
for some arbitrary choice of $i_{0}\in\mathsf{I}$. For $i\neq i_{0}$,
\[
\frac{\left(a_{i}\middle|P_{j}\middle|\rho\right)}{\left(u\middle|P_{j}\middle|\rho\right)}=\frac{\left(u\middle|P_{i}P_{j}\middle|\rho\right)}{\left(u\middle|P_{j}\middle|\rho\right)}=\frac{\delta_{ij}}{\left(u\middle|P_{i}\middle|\rho\right)}\left(u\middle|P_{i}\middle|\rho\right)=\delta_{ij},
\]
where we have used the definition of SRT. Finally, for $i=i_{0}$,
\[
\frac{\left(a_{i_{0}}\middle|P_{j}\middle|\rho\right)}{\left(u\middle|P_{j}\middle|\rho\right)}=\frac{\left(u\middle|P_{i_{0}}P_{j}\middle|\rho\right)}{\left(u\middle|P_{j}\middle|\rho\right)}+\sum_{i\notin\mathsf{I}}\frac{\left(u\middle|P_{i}P_{j}\middle|\rho\right)}{\left(u\middle|P_{j}\middle|\rho\right)}=\delta_{i_{0}j}+\sum_{i\notin\mathsf{I}}\delta_{ij}\frac{\left(u\middle|P_{i}\middle|\rho\right)}{\left(u\middle|P_{j}\middle|\rho\right)},
\]
but the second term always vanishes because $P_{i}\rho=0$ for $i\notin\mathsf{I}$.
\end{proof}

\subsection{The general form of objective states in causal theories\label{subsec:form objective}}

It is not hard to show that if the joint state in the OG is SBS, the
players can win the game. Indeed, suppose the joint state is
\begin{equation}
\rho_{\mathrm{SE}}=\sum_{i=1}^{r}p_{i}\alpha_{i,\mathrm{S}}\otimes\rho_{i,\mathrm{E}_{1}}\otimes\ldots\otimes\rho_{i,\mathrm{E}_{n}},\label{eq:SBS}
\end{equation}
with $p_{i}>0$, where $\mathrm{E}$ denotes the joint environment
composed of all the fragments controlled by the players $\mathrm{E}=\mathrm{E}_{1}\ldots\mathrm{E}_{n}$.
Note that this state respects the strong independence condition: the
states of the various players are only correlated by index $i$ labeling
the outcome found by the referee on $\mathrm{S}$. In this game, the
referee applies the SRT containing the transformations $\left|\alpha_{i}\right)\left(a_{i}\right|$.
If $\left\{ \alpha_{i}\right\} _{i=1}^{r}$ is a maximal set of perfectly
distinguishable pure states, $\left\{ \left|\alpha_{i}\right)\left(a_{i}\right|\right\} _{i=1}^{r}$
will be a test, otherwise it is enough to add other pure states to
$\left\{ \alpha_{i}\right\} _{i=1}^{r}$ until it becomes maximal
$\left\{ \alpha_{i}\right\} _{i=1}^{d}$, with $d>r$. In this latter
case, the SRT performed by the referee will be $\left\{ \left|\alpha_{i}\right)\left(a_{i}\right|\right\} _{i=1}^{d}$.

What about the other players? What is their strategy to win the game?
Since the states $\left\{ \rho_{i,\mathrm{E}_{k}}\right\} _{i=1}^{r}$
are perfectly distinguishable for every $k$, each player just needs
to perform the SRT associated with them, namely $\left\{ P_{i,\mathrm{E}_{k}}=\left|\rho_{i,\mathrm{E}_{k}}\right)\left(a_{i,\mathrm{E}_{k}}\right|\right\} _{i=1}^{r}$,
where $\left\{ a_{i,\mathrm{E}_{k}}\right\} _{i=1}^{r}$ is the measurement
that distinguishes them. Note that, $P_{i,\mathrm{E}_{k}}\rho_{j,\mathrm{E}_{k}}=\delta_{ij}\rho_{i,\mathrm{E}_{k}}$,
so this SRT does not disturb the state~\eqref{eq:SBS} in a strong
sense. This shows that every causal theory has objective states.

The non-trivial part is to show that these are the \emph{only} objective
states. The key step is the following lemma.
\begin{lem}
\label{lem:non-disturbing}Let $\rho_{\mathrm{SE}}$ be a state such
that $\mathrm{tr}_{\mathrm{E}}\rho_{\mathrm{SE}}=\sum_{i=1}^{r}p_{i}\alpha_{i}$,
where $p_{i}>0$ for all $i$, and the $\alpha_{i}$'s are the pure
states of a $d$-dimensional classical set $\boldsymbol{\alpha}$,
with $d\geq r$. If 
\begin{equation}
\sum_{i=1}^{d}\left(\left|\alpha_{i}\right)\left(a_{i}\right|\otimes P_{i}\right)\rho=\rho,\label{eq:condition}
\end{equation}
where the $P_{i}$'s are transformations in an SRT on $\mathrm{E}$,
then $\rho_{\mathrm{SE}}$ must be of the form
\[
\rho=\sum_{i=1}^{r}p_{i}\alpha_{i}\otimes\rho_{i},
\]
where $\left\{ \rho_{i}\right\} _{i=1}^{r}$ are perfectly distinguishable
states of $\mathrm{E}$.
\end{lem}
\begin{proof}
Let us rewrite eq.\ \eqref{eq:condition} in diagrams.\begin{equation}\label{eq:condition diagrams} \sum_{i=1}^{d}\!\!\!\!\begin{aligned}\Qcircuit @C=1em @R=.7em @!R { & \multiprepareC{1}{\rho} & \qw \poloFantasmaCn{\rS} &\measureD{a_i} &\prepareC{\alpha_i} & \qw \poloFantasmaCn{\rS} &\qw \\ & \pureghost{\rho} & \qw \poloFantasmaCn{\rE} & \gate{P_i} & \qw \poloFantasmaCn{\rE} &\qw &\qw}\end{aligned}~=\!\!\!\!\begin{aligned}\Qcircuit @C=1em @R=.7em @!R { & \multiprepareC{1}{\rho} & \qw \poloFantasmaCn{\rS} &\qw \\ & \pureghost{\rho} & \qw \poloFantasmaCn{\rE} &\qw}\end{aligned}~. \end{equation}The
left-hand side can be rewritten as\[ \sum_{i=1}^{d}\!\!\!\!\begin{aligned}\Qcircuit @C=1em @R=.7em @!R { & \multiprepareC{1}{\rho} & \qw \poloFantasmaCn{\rS} &\measureD{a_i} &\prepareC{\alpha_i} & \qw \poloFantasmaCn{\rS} &\qw \\ & \pureghost{\rho} & \qw \poloFantasmaCn{\rE} & \gate{P_i} & \qw \poloFantasmaCn{\rE} &\qw &\qw}\end{aligned}~=~\sum_{i=1}^{d}\lambda_i\!\!\!\!\begin{aligned}\Qcircuit @C=1em @R=.7em @!R { &\prepareC{\alpha_i} & \qw \poloFantasmaCn{\rS} &\qw \\ & \prepareC{\rho_i} & \qw \poloFantasmaCn{\rE} &\qw}\end{aligned}~, \]where
we have set\begin{equation}\label{eq:rho_i} \lambda_i\!\!\!\!\begin{aligned}\Qcircuit @C=1em @R=.7em @!R {& \prepareC{\rho_i} & \qw \poloFantasmaCn{\rE} &\qw}\end{aligned}~:=\!\!\!\!\begin{aligned}\Qcircuit @C=1em @R=.7em @!R { & \multiprepareC{1}{\rho} & \qw \poloFantasmaCn{\rS} &\qw &\qw &\measureD{a_i} \\ & \pureghost{\rho} & \qw \poloFantasmaCn{\rE} & \gate{P_i} & \qw \poloFantasmaCn{\rE}&\qw}\end{aligned}~, \end{equation}so
that $\rho_{i}$ is normalized, and $\lambda_{i}\in\left[0,1\right]$.
Now eq.~\eqref{eq:condition diagrams} becomes\begin{equation}\label{eq:resolved} \begin{aligned}\Qcircuit @C=1em @R=.7em @!R { & \multiprepareC{1}{\rho} & \qw \poloFantasmaCn{\rS} &\qw \\ & \pureghost{\rho} & \qw \poloFantasmaCn{\rE} &\qw}\end{aligned}~=~\sum_{i=1}^{d}\lambda_i\!\!\!\!\begin{aligned}\Qcircuit @C=1em @R=.7em @!R { &\prepareC{\alpha_i} & \qw \poloFantasmaCn{\rS} &\qw \\ & \prepareC{\rho_i} & \qw \poloFantasmaCn{\rE} &\qw}\end{aligned}~. \end{equation}To
conclude the proof we must show that the non-vanishing $\lambda_{i}$'s
are the coefficients $p_{i}$'s, and that the states $\rho_{i}$'s
are perfectly distinguishable. Since $\mathrm{tr}_{\mathrm{E}}\rho_{\mathrm{SE}}=\sum_{i=1}^{r}p_{i}\alpha_{i}$,
we have that for $i=1,\ldots,r$\begin{equation}\label{eq:p_i} p_i~=\!\!\!\!\begin{aligned}\Qcircuit @C=1em @R=.7em @!R { & \multiprepareC{1}{\rho} & \qw \poloFantasmaCn{\rS} &\measureD{a_i} \\ & \pureghost{\rho} & \qw \poloFantasmaCn{\rE} & \measureD{u}}\end{aligned}~. \end{equation}By
eq.~\eqref{eq:resolved}, one has\[ \lambda_i~=\!\!\!\begin{aligned}\Qcircuit @C=1em @R=.7em @!R { & \multiprepareC{1}{\rho} & \qw \poloFantasmaCn{\rS} &\measureD{a_i} \\ & \pureghost{\rho} & \qw \poloFantasmaCn{\rE} & \measureD{u}}\end{aligned}~=~p_i \]for
$i=1,\ldots,r$, and $\lambda_{i}=0$ for $i>r$. This means we can
replace the summation from 1 to $d$ with a summation from 1 to $r$.

Now we must prove that the states $\rho_{i}$ are perfectly distinguishable.
Now rewrite eq.~\eqref{eq:rho_i}, for $i=1,\ldots,r$ as 
\[
p_{i}\rho_{i}=\mu_{i}P_{i}\sigma_{i},
\]
where\[ \mu_i\!\!\!\!\begin{aligned}\Qcircuit @C=1em @R=.7em @!R {& \prepareC{\sigma_i} & \qw \poloFantasmaCn{\rE} &\qw}\end{aligned}~:=\!\!\!\!\begin{aligned}\Qcircuit @C=1em @R=.7em @!R { & \multiprepareC{1}{\rho} & \qw \poloFantasmaCn{\rS} &\measureD{a_i} \\ & \pureghost{\rho} & \qw \poloFantasmaCn{\rE} & \qw }\end{aligned}~, \]$\sigma_{i}$
is normalized and $\mu_{i}\in\left[0,1\right]$. A quick comparison
with eq.~\eqref{eq:p_i} shows that $\mu_{i}=p_{i}>0$ for all $i=1,\ldots,r$,
thus $\rho_{i}=P_{i}\sigma_{i}$. Lemma~\ref{lem:distinguishable SRM}
ensures that the states $\rho_{i}$ are perfectly distinguishable
(and in this case without even renormalizing them).
\end{proof}
Now we can give the proof of our main result (cf.\ theorem~\ref{fig:The-state-space}
in the main text).
\begin{proof}[Proof of theorem~\ref{fig:The-state-space} in the main text]
In the setting of the OG, if the referee checks the findings with
the SRT $\left\{ \left|\alpha_{i}\right)\left(a_{i}\right|\right\} _{i=1}^{d}$
on $\mathrm{S}$, and the other players apply some SRTs $\left\{ P_{j_{k},\mathrm{E}_{k}}\right\} $
on each $\mathrm{E}_{k}$, the probability of a joint outcome $\left(i,j_{1},\ldots,j_{n}\right)$
is 
\[
p_{ij_{1}\ldots j_{n}}=\mathrm{tr}\left[\left(\left|\alpha_{i}\right)\left(a_{i}\right|\otimes P_{j_{1},\mathrm{E}_{1}}\otimes\ldots\otimes P_{j_{n},\mathrm{E}_{n}}\right)\rho_{\mathrm{SE}}\right].
\]
Imposing the agreement condition we must have $p_{ij_{1}\ldots j_{n}}=0$
unless $i=j_{1}=\ldots=j_{n}$ \citep{Objectivity}. This means that,
if we forget the outcome, the state after the measurement is 
\[
\sum_{i,j_{1},\ldots,j_{n}}\left(\left|\alpha_{i}\right)\left(a_{i}\right|\otimes P_{j_{1},\mathrm{E}_{1}}\otimes\ldots\otimes P_{j_{n},\mathrm{E}_{n}}\right)\rho_{\mathrm{SE}}=\sum_{i=1}^{d}\left(\left|\alpha_{i}\right)\left(a_{i}\right|\otimes P_{i,\mathrm{E}_{1}}\otimes\ldots\otimes P_{i,\mathrm{E}_{n}}\right)\rho_{\mathrm{SE}}.
\]
Now, let us define $P_{i}:=P_{i,\mathrm{E}_{1}}\otimes\ldots\otimes P_{i,\mathrm{E}_{n}}$,
which is an SRT on $\mathrm{E}$. Imposing the Bohr non-disturbance
condition (cf.\ definition~\ref{def:non-disturbing} in main text)
to the test $\left\{ \left|\alpha_{i}\right)\left(a_{i}\right|\otimes P_{i}\right\} $,
we find
\[
\sum_{i=1}^{d}\left(\left|\alpha_{i}\right)\left(a_{i}\right|\otimes P_{i,\mathrm{E}}\right)\rho_{\mathrm{SE}}=\rho_{\mathrm{SE}}.
\]
Now we are in the situation of lemma~\ref{lem:non-disturbing}, so
we know that $\rho_{\mathrm{SE}}$ must be of the form $\rho=\sum_{i=1}^{r}p_{i}\alpha_{i}\otimes\rho_{i}$,
where the $\rho_{i}$'s are perfectly distinguishable states of $\mathrm{E}$.
Imposing the strong independence condition, $\rho_{i}$ must be a
product state, with the only correlations given by the index $i$:
\[
\rho_{i}=\rho_{i,\mathrm{E}_{1}}\otimes\ldots\otimes\rho_{i,\mathrm{E}_{n}},
\]
where for every $k$, the states $\left\{ \rho_{i,\mathrm{E}_{k}}\right\} _{i=1}^{r}$
are perfectly distinguishable. This concludes the proof.
\end{proof}

\section{Emergence of composite classical systems\label{sec:Emergence-of-composite}}

In this appendix we elaborate more on the issue of composition of
classical sub-theories of a given causal theory. For simplicity, we
focus on bipartite systems; the generalization to more than two parties
will be straightforward.

Consider now a bipartite system $\mathrm{AB}$ of a generic causal
theory, and suppose system $\mathrm{A}$ has the classical set $\boldsymbol{\alpha}$,
and system $\mathrm{B}$ the classical set $\boldsymbol{\beta}$.
If $\boldsymbol{\alpha}$ is to represent an actual classical sub-theory
for system $\mathrm{A}$, and $\boldsymbol{\beta}$ an actual classical
sub-theory for system $\mathrm{B}$, it is natural to expect that
the classical set for the composite system should mirror the composition
of classical theory. Consequently, we would like to define the composite
classical set as
\begin{equation}
\boldsymbol{\alpha\beta}:=\mathrm{Conv}\left\{ \gamma_{\mathrm{A}}\otimes\gamma_{\mathrm{B}}:\gamma_{\mathrm{A}}\in\boldsymbol{\alpha},\gamma_{\mathrm{B}}\in\boldsymbol{\beta}\right\} .\label{eq:composite classical}
\end{equation}
In particular, this definition implies that the pure states of the
classical set for $\mathrm{AB}$ should \emph{all} be of the form
$\alpha_{i}\otimes\beta_{j}$, where $\alpha_{i}$ is a pure state
of $\boldsymbol{\alpha}$, and $\beta_{j}$ a pure state of $\boldsymbol{\beta}$.
However, here we face two problems. The first is that the product
of two pure states may not be pure in general, as shown in ref.~\citep{BarnumGraydonWilceCCEJA},
and the second is that the set $\left\{ \alpha_{i}\otimes\beta_{j}\right\} $
may not be maximal for the composite system $\mathrm{AB}$ as shown
in ref.~\citep{Hardy-limited}, which means that there are extra
pure states to add.

Now, in this setting, axioms~\ref{axm:product pure} and \ref{axm:information locality}
in the main text are introduced to rule out this pathological \emph{holistic}
behavior. Indeed, if the first axiom fails, and the product of two
pure states is \emph{not} pure, the idea that the classical states
of a composite system be reducible to the classical states of its
components faces a considerable difficulty. In this case, since the
product states are mixed, the theory is so holistic that, to construct
the classical set for the composite system, we have to look for completely
different states. If, instead, the theory satisfies axiom~\ref{axm:product pure}
in the main text, but it fails axiom~\ref{axm:information locality}
in the main text, we can construct the classical set for $\mathrm{AB}$
partially out of $\boldsymbol{\alpha}$ and $\boldsymbol{\beta}$,
but we need some extra pure states of $\mathrm{AB}$ to make it maximal.
Even in this case the theory shows a holistic behavior, and does not
support the emergence of proper classical composite systems. If both
axioms are satisfied, we have that the classical set for the composite
system $\mathrm{AB}$ is given by~\eqref{eq:composite classical}.

In a similar spirit, in the presence of axioms~\ref{axm:product pure}
and \ref{axm:information locality} in the main text, it is natural
to expect that the decoherence process on a bipartite system is reducible
to the decoherence of the two components \citep{Selby-leaks,Selby-entanglement2,CPT,2roads}.
In formula,
\begin{equation}
D_{\boldsymbol{\alpha\beta}}\rho_{\mathrm{AB}}=\left(D_{\boldsymbol{\alpha}}\otimes D_{\boldsymbol{\beta}}\right)\rho_{\mathrm{AB}}\label{eq:decoherence-product}
\end{equation}
for every bipartite state $\rho_{\mathrm{AB}}$. However, using the
most general definition of complete decoherence (definition~\ref{def:complete decoherence}
in the main text) we cannot compare the action of $D_{\boldsymbol{\alpha}}\otimes D_{\boldsymbol{\beta}}$
to the action of $D_{\boldsymbol{\alpha\beta}}$, as there is no specific
recipe for decohering states.

Let us see, instead, what happens when we consider TIDs. Let $\left\{ a_{i}\right\} _{i=1}^{d_{\mathrm{A}}}$
and $\left\{ b_{j}\right\} _{j=1}^{d_{\mathrm{B}}}$ be the measurements
associated with $\boldsymbol{\alpha}$ and $\boldsymbol{\beta}$ respectively;
the measurement associated with $\boldsymbol{\alpha\beta}$ will be
$\left\{ a_{i}\otimes b_{j}\right\} _{i=1,}^{d_{\mathrm{A}}}\phantom{}_{j=1}^{d_{\mathrm{B}}}$.
Therefore, the TID $\widehat{D}_{\boldsymbol{\alpha\beta}}$ is the
channel
\[
\widehat{D}_{\boldsymbol{\alpha\beta}}=\sum_{i=1}^{d_{\mathrm{A}}}\sum_{j=1}^{d_{\mathrm{B}}}\left|\alpha_{i}\right)\left|\beta_{j}\right)\left(a_{i}\right|\left(b_{j}\right|.
\]
An easy rewriting shows that 
\[
\widehat{D}_{\boldsymbol{\alpha\beta}}=\sum_{i=1}^{d_{\mathrm{A}}}\left|\alpha_{i}\right)\left(a_{i}\right|\otimes\sum_{j=1}^{d_{\mathrm{B}}}\left|\beta_{j}\right)\left(b_{j}\right|=\widehat{D}_{\boldsymbol{\alpha}}\otimes\widehat{D}_{\boldsymbol{\beta}},
\]
which is one of the desiderata of ref.~\citep{Selby-entanglement2}.
This means that for TIDs the behavior on composite systems is fully
reducible to the behavior on the components. Note that this result
is stronger than eq.~\eqref{eq:decoherence-product} in the absence
of Local Tomography \citep{Chiribella-purification} (see appendix~\ref{sec:General-probabilistic-theories}). 
\end{document}

%% file: myQcircuit.tex
\usepackage[matrix,frame,arrow]{xy}
\usepackage{amsmath}

\newcommand{\qw}[1][-1]{\ar @{-} [0,#1]}

\newcommand{\gate}[1]{*{\xy *+<.6em>{#1};p\save+LU;+RU **\dir{-}\restore\save+RU;+RD **\dir{-}\restore\save+RD;+LD **\dir{-}\restore\POS+LD;+LU **\dir{-}\endxy} \qw}

\newcommand{\measureD}[1]{*{\xy*+=+<.5em>{\vphantom{\rule{0em}{.1em}#1}}*\cir{r_l};p\save*!R{#1} \restore\save+UC;+UC-<.5em,0em>*!R{\hphantom{#1}}+L **\dir{-} \restore\save+DC;+DC-<.5em,0em>*!R{\hphantom{#1}}+L **\dir{-} \restore\POS+UC-<.5em,0em>*!R{\hphantom{#1}}+L;+DC-<.5em,0em>*!R{\hphantom{#1}}+L **\dir{-} \endxy} \qw}

\newcommand{\multimeasureD}[2]{*+<1em,.9em>{\hphantom{#2}}\save[0,0].[#1,0];p\save !C *{#2},p+LU+<0em,0em>;+RU+<-.8em,0em> **\dir{-}\restore\save +LD;+LU **\dir{-}\restore\save +LD;+RD-<.8em,0em> **\dir{-} \restore\save +RD+<0em,.8em>;+RU-<0em,.8em> **\dir{-} \restore \POS !UR*!UR{\cir<.9em>{r_d}};!DR*!DR{\cir<.9em>{d_l}}\restore \qw}

\newcommand{\multigate}[2]{*+<1em,.9em>{\hphantom{#2}} \qw \POS[0,0].[#1,0];p !C *{#2},p \save+LU;+RU **\dir{-}\restore\save+RU;+RD **\dir{-}\restore\save+RD;+LD **\dir{-}\restore\save+LD;+LU **\dir{-}\restore}
\newcommand{\ghost}[1]{*+<1em,.9em>{\hphantom{#1}} \qw}
\newcommand{\Qcircuit}[1][0em]{\xymatrix @*=<#1>} 

\newcommand{\pureghost}[1]{*+<1em,.9em>{\hphantom{#1}}}
\newcommand{\multiprepareC}[2]{*+<1em,.9em>{\hphantom{#2}}\save[0,0].[#1,0];p\save !C
  *{#2},p+RU+<0em,0em>;+LU+<+.8em,0em> **\dir{-}\restore\save +RD;+RU **\dir{-}\restore\save
  +RD;+LD+<.8em,0em> **\dir{-} \restore\save +LD+<0em,.8em>;+LU-<0em,.8em> **\dir{-} \restore \POS
  !UL*!UL{\cir<.9em>{u_r}};!DL*!DL{\cir<.9em>{l_u}}\restore}
\newcommand{\prepareC}[1]{*{\xy*+=+<.5em>{\vphantom{#1\rule{0em}{.1em}}}*\cir{l^r};p\save*!L{#1} \restore\save+UC;+UC+<.5em,0em>*!L{\hphantom{#1}}+R **\dir{-} \restore\save+DC;+DC+<.5em,0em>*!L{\hphantom{#1}}+R **\dir{-} \restore\POS+UC+<.5em,0em>*!L{\hphantom{#1}}+R;+DC+<.5em,0em>*!L{\hphantom{#1}}+R **\dir{-} \endxy}}
\newcommand{\poloFantasmaCn}[1]{{{}^{#1}_{\phantom{#1}}}}